\documentclass[10pt, a4paper, english]{article}

\usepackage{lmodern}
\usepackage[T1]{fontenc}
\usepackage[utf8]{inputenc}
\usepackage[english]{babel}
\usepackage{microtype}
\usepackage{csquotes}

\usepackage[inline]{enumitem}
\setlist[itemize, 1]{label=\textbullet}
\setlist[itemize, 2]{label=$\circ$}
\setlist[itemize, 3]{label={\tiny$\blacksquare$}}
\setlist[itemize, 4]{label={\tiny$\square$}}
\usepackage{moreenum}

\usepackage{amsthm}
\numberwithin{equation}{section}
\usepackage[
	x11names*,
	svgnames*,
	dvipsnames,
]{xcolor}

\definecolor{hyperlinkcolor}{RGB}{11,0,128} 

\usepackage{amsmath}
\usepackage{amssymb}
\usepackage{mathtools}
\usepackage[
	classfont = sanserif,
	langfont  = caps,
	funcfont  = roman,
	full,
	disableredefinitions,
]{complexity}
\usepackage[heavycircles]{stmaryrd}

\usepackage{scalerel}
\usepackage{graphicx}

\usepackage{tikz}
\usetikzlibrary{calc}
\tikzset{>=stealth} 

\usepackage[%
	labelfont = {normalfont},%
]{caption}
\usepackage{subcaption}

\numberwithin{table}{section}
\numberwithin{figure}{section}

\usepackage[
	backend     = biber,
	style       = alphabetic,
	sorting     = nty,
	maxbibnames = 99,
	backref     = true,
]{biblatex}
\usepackage{authblk}
\usepackage{multicol}
\usepackage{geometry}
\newgeometry{
	hcentering,
	marginparwidth = 2.6cm,
	headheight     = 13.6pt,
} 

\usepackage{fancyhdr}
\usepackage{etoolbox}

\usepackage{hyperref}

\hypersetup{
    plainpages  = false,
    linktoc     = section,
    colorlinks  = true,
    breaklinks  = true,
    pageanchor  = true,
    linkcolor   = hyperlinkcolor,
    anchorcolor = black,
    citecolor   = hyperlinkcolor,
    filecolor   = hyperlinkcolor,
    menucolor   = hyperlinkcolor,
    runcolor    = hyperlinkcolor,
    urlcolor    = hyperlinkcolor,
    pdftitle    = {The flattening operator and implications in team-based logics},
    pdfauthor   = {Arnaud Durand, Juha Kontinen, Werner Merian and Jouko Väänänen},
} 

\usepackage[%
	symbol       = $\uparrow\;$,%
	numberlinked = false,%
]{footnotebackref}

\usepackage[
	noabbrev,
]{cleveref}

\crefname{part}{part}{parts}
\crefname{chapter}{chapter}{chapters}
\crefname{section}{section}{sections}
\crefname{subsection}{subsection}{subsections}
\crefname{subsubsection}{subsubsection}{subsubsections}
\crefname{paragraph}{paragraph}{paragraphs}
\crefname{subparagraph}{subparagraph}{subparagraphs}

\crefname{figure}{figure}{figures}
\crefname{subfigure}{subfigure}{subfigures}

\crefname{table}{table}{tables}
\crefname{subtable}{subtable}{subtables}

\crefname{footnote}{footnote}{footnotes}
\crefname{marginenote}{marginenote}{marginenotes}
\crefname{endnote}{endnote}{endnotes}

\crefname{item}{item}{items}

\crefname{equation}{equation}{equations}
\crefname{condition}{condition}{conditions} 

\makeatletter
\def \@makefntext#1{%
	\parindent 1em%
	\noindent%
	\hb@xt@ 1.8em{\hss \hbox{\normalfont \@thefnmark.\,\,}}%
	#1%
}

\def \p@cref#1{%
	\expandafter\ifx\csname p@cref@#1\endcsname\relax%
	\expandafter\csname p@#1\endcsname%
	\else%
	\expandafter\csname p@cref@#1\endcsname%
	\fi%
}

\def \thecref#1{%
	\expandafter\ifx\csname thecref#1\endcsname\relax%
	\expandafter\csname the#1\endcsname%
	\else%
	\expandafter\csname thecref#1\endcsname%
	\fi%
}

\def\setcrefcounterprefix#1{\expandafter\def\csname p@cref@#1\endcsname##1}
\def\setcrefthecounter#1{\expandafter\def\csname thecref#1\endcsname##1}

\patchcmd\refstepcounter@noarg%
{%
	\csname p@#1\endcsname\csname the#1\endcsname%
}{%
	\p@cref{#1}\thecref{#1}%
}%
{\typeout{\string\refstepcounter@noarg\space patch success for cleveref}}%
{\typeout{\string\refstepcounter@noarg\space patch failed for cleveref}}

\patchcmd\refstepcounter@optarg%
{%
	\csname p@#2\endcsname\csname the#2\endcsname%
}{%
	\p@cref{#2}\thecref{#2}%
}%
{\typeout{\string\refstepcounter@optarg\space patch success for cleveref}}%
{\typeout{\string\refstepcounter@optarg\space patch failed for cleveref}}

\patchcmd\H@refstepcounter%
{%
	\csname p@#1\endcsname\csname the#1\endcsname%
}{%
	\p@cref{#1}\thecref{#1}%
}%
{\typeout{\string\H@refstepcounter\space patch success for cleveref}}%
{\typeout{\string\H@refstepcounter\space patch failed for cleveref}}

\newcommand{\SetEnumerateInTheoremEnv}[1]{%
	\AtBeginEnvironment{#1}{%
		\setcrefcounterprefix{enumi}{\thecref{#1}.##1}%
		\crefalias{enumi}{#1}%
		\setlist[enumerate,1]{%
			label = {\normalfont(\textit{\roman*})},%
			ref   = (\textit{\roman*}),%
		}%
	}%
}

\newcommand{\crefnameabbr}[3]{%
	\expandafter\gdef\csname crefnameabbr@#1\endcsname{{#1}{#2}{#3}}%
}

\newcommand{\crefabbr}[1]{%
	\begingroup%
		\expandafter\ifx\csname r@#1@cref\endcsname\relax%
			\protect\G@refundefinedtrue%
			\nfss@text{\reset@font\bfseries ??}%
			\@latex@warning{Reference `#1' on page \thepage \space undefined}%
		\else%
			\cref@gettype{#1}{\@type}%
			\edef\@args{\csname crefnameabbr@\@type\endcsname}%
			\expandafter\crefname\@args%
			\cref{#1}%
		\fi%
	\endgroup%
}
\makeatother

\newtheoremstyle{mydefinitionstyle}%
{}%
{}%
{}%
{}%
{\bfseries}%
{.}%
{ }%
{\thmname{#1}\thmnumber{ #2}\thmnote{ (#3)}}%

\theoremstyle{mydefinitionstyle}

\newtheorem{definition}{Definition}[section]
\crefname{definition}{definition}{definitions}
\crefnameabbr{definition}{def.}{defs.}

\newtheorem{notation}[definition]{Notation}
\crefname{notation}{notation}{notations}
\crefnameabbr{notation}{nota.}{notas.}

\newtheorem{remark}[definition]{Remark}
\crefname{remark}{remark}{remarks}
\crefnameabbr{remark}{rmk.}{rmks.}

\crefname{example}{example}{examples}
\crefnameabbr{example}{exm.}{exms.}

\newtheorem*{example*}{Example}

\crefname{exercise}{exercise}{exercises}
\crefnameabbr{exercise}{exr.}{exrs.}

\newtheoremstyle{mytheoremstyle} %
{}%
{}%
{\slshape}%
{}%
{\bfseries}%
{.}%
{ }%
{\thmname{#1}\thmnumber{ #2}\thmnote{ (#3)}}%

\theoremstyle{mytheoremstyle}

\crefname{claim}{claim}{claims}
\crefnameabbr{claim}{clm.}{clms.}
\SetEnumerateInTheoremEnv{claim}

\crefname{fact}{fact}{facts}
\crefnameabbr{fact}{fct.}{fcts.}
\SetEnumerateInTheoremEnv{fact}

\newtheorem{property}[definition]{Property}
\crefname{property}{property}{properties}
\crefnameabbr{property}{pty.}{pties.}
\SetEnumerateInTheoremEnv{property}

\newtheorem{proposition}[definition]{Proposition}
\crefname{proposition}{proposition}{propositions}
\crefnameabbr{proposition}{prop.}{props.}
\SetEnumerateInTheoremEnv{proposition}

\newtheorem{lemma}[definition]{Lemma}
\crefname{lemma}{lemma}{lemmas}
\crefnameabbr{lemma}{lem.}{lems.}
\SetEnumerateInTheoremEnv{lemma}

\newtheorem{theorem}[definition]{Theorem}
\crefname{theorem}{theorem}{theorems}
\crefnameabbr{theorem}{thm.}{thms.}
\SetEnumerateInTheoremEnv{theorem}

\newtheorem{corollary}[definition]{Corollary}
\crefname{corollary}{corollary}{corollaries}
\crefnameabbr{corollary}{cor.}{cors.}
\SetEnumerateInTheoremEnv{corollary}

\crefname{axiom}{axiom}{axioms}
\crefnameabbr{axiom}{ax.}{axs.}
\SetEnumerateInTheoremEnv{axiom}

\crefname{principle}{principle}{principles}
\crefnameabbr{principle}{ppl.}{ppls.}
\SetEnumerateInTheoremEnv{principle}

\crefname{conjecture}{conjecture}{conjectures}
\crefnameabbr{conjecture}{conj.}{conjs.}
\SetEnumerateInTheoremEnv{conjecture}

\crefname{hypothesis}{hypothesis}{hypotheses}
\crefnameabbr{hypothesis}{hyp.}{hyps.}
\SetEnumerateInTheoremEnv{hypothesis}

\newcommand{\latinexpr}[1]{\textit{#1}}
\newcommand{\newlatinexpr}[2]{\newcommand{#1}{\latinexpr{#2}}}
\newlatinexpr{\ie}{i.e.}
\newlatinexpr{\eg}{e.g.}

\newlatinexpr{\apriori}{a priori}

\newcommand{\email}[1]{\href{mailto:#1}{\texttt{#1}}}

\newcommand{\partexpr}{\emph}
\newcommand{\defexpr}{\emph}

\newcommand{\emphexpr}{\emph}
\newcommand{\announceexpr}{\emph}

\newcommand{\quoexpr}[1]{``#1''}
\newcommand{\extexpr}[1]{``#1''}
\newcommand{\citeexpr}[1]{``#1''} 

\newcommand{\nequiv}{\not\equiv} 

\newcommand{\tuple}{\vec}

\newcommand{\replaces}{\middle\slash}

\newcommand{\defeq}{\overset{{\normalfont\text{def}}}{=}}

\newcommand{\emptyarg}{\cdot}

\DeclareMathOperator{\dom}{dom}

\newcommand{\powerset}{\mathcal{P}}
\newcommand{\powersetstar}{\powerset^{*}}

\newcommand{\setsep}{\;\middle\vert\;} %

\newcommand{\syntacticeq}{=}
\newcommand{\syntacticneq}{\neq} %

\DeclareMathOperator{\ar}{ar}

\DeclareMathOperator{\Vr}{Vr}
\DeclareMathOperator{\FV}{FV}

\DeclareMathOperator{\Aut}{Aut} 

\newclass{\LOGSPACE}{LOGSPACE}
\newclass{\NLOGSPACE}{NLOGSPACE} %

\newcommand{\fragmentfont}[1]{{\normalfont\textsf{#1}}}

\DeclareDocumentCommand \newfragment{ m m }{%
	\DeclareDocumentCommand{#1}{ o }{%
		\fragmentfont{#2}%
		\IfValueT {##1}%
		{_{##1}}%
	}
}

\newfragment{\SO}{SO}
\newfragment{\MSO}{MSO}
\newfragment{\ESO}{ESO}

\newfragment{\EMSO}{EMSO}
\newfragment{\ESOHORN}{ESO-Horn}

\newfragment{\GFP}{GFP}
\newfragment{\LFP}{LFP}
\newcommand{\GFPpos}{\GFP^{\fragmentfont{+}}}

\newcommand{\dep}[1][]{{=_{#1}\mkern-1.2mu}}
\newcommand{\anon}[1][]{\mathbin{\Upsilon_{#1}}}

\newcommand{\incl}[1][]{\mathbin{\subseteq_{#1}}}
\newcommand{\excl}[1][]{\mathbin{\vert_{#1}}}

\newcommand{\ind}{\mathbin{\bot}}

\newcommand{\cneg}{{\sim}} %

\newcommand{\regularimplication}{\hookrightarrow}

\DeclareMathOperator*{\bigovee}{\scalerel*{\ovee}{\sum}}

\DeclareMathOperator{\TeamCl}{Cl} %

\newcommand{\vDashTarski}{\vDash}%
\newcommand{\vDashteam}{\vDash}%
\newcommand{\nvDashteam}{\nvDash}%

\newcommand{\f}{\mathsf{f}}%

\newcommand{\F}{\mathsf{F}}%

\newfragment{\DEP}{DEP}
\newfragment{\ANON}{ANON}

\newfragment{\INC}{INC}
\newfragment{\EXC}{EXC}
\newfragment{\INEX}{INEX}

\newfragment{\IND}{IND}

\newfragment{\NEG}{FON}

\newcommand{\consformat}{\dep(\cdot)}
\newcommand{\depformat}[1][]{\dep[#1](\cdot,\cdot)}

\newcommand{\FOanon}[1][]{\ensuremath{\FO(\anon[#1])}}

\newcommand{\FOcons}{\ensuremath{\FO(\consformat)}}
\newcommand{\FOdep}[1][]{\ensuremath{\FO(\depformat[#1])}}

\newcommand{\FOincl}[1][]{\ensuremath{\FO(\incl[#1])}}
\newcommand{\FOexcl}[1][]{\ensuremath{\FO(\excl[#1])}}

\newcommand{\FOind}{\ensuremath{\FO(\ind)}}

\newcommand{\FOneg}{\ensuremath{\FO(\cneg)}}

\title{The flattening operator in team-based logics}

\author[1,*]{Arnaud Durand}
\author[2,\dag]{Juha Kontinen}
\author[1,\ddag]{Werner Mérian}
\author[2,3,\S]{Jouko Väänänen}

\affil[1]{Université Paris Cité, CNRS, IMJ-PRG, Paris, France}
\affil[2]{Department of Mathematics and Statistics, University of Helsinki, Finland}
\affil[3]{Institute for Logic, Language and Computation, University of Amsterdam, The Netherlands}
\affil[ ]{%
	\textsuperscript{*}\email{durand@imj-prg.fr},
	\textsuperscript{\dag}\email{juha.kontinen@helsinki.fi},
	\textsuperscript{\ddag}\email{merian@imj-prg.fr},
	\textsuperscript{\S}\email{jouko.vaananen@helsinki.fi}
}

\date{May 19, 2025}

\begin{document}
	\pagestyle{fancy}
	
	\maketitle
	
	\begin{abstract}
		We propose a systematic study of the so-called \announceexpr{flattening operator} in team semantics. This operator was first introduced by Hodges in 1997, and has not been studied in more detail since. We begin a systematic study of the expressive power this operator adds to the most well-known team-based logics, such as \announceexpr{dependence logic}, \announceexpr{anonymity logic}, \announceexpr{inclusion logic} and \announceexpr{exclusion logic}.
	\end{abstract}

\section{Introduction}

Team semantics is a mathematical framework for studying concepts and phenomena that arise in the presence of plurality of data. Examples of such concepts are, for example,  functional dependence in database theory and conditional independence of random variables in statistics. The beginning of the area can be traced back to the introduction of dependence logic in \cite{Vaananen:2007:dependence}. In dependence logic, formulas are interpreted by sets of assignments (called \announceexpr{teams}) and it extends the syntax of first-order logic by dependence atoms $\dep(x,y)$ expressing that the values of the variables $x$ functionally determine the values of $y$ in a team. 

Since the introduction of dependence logic, the expressivity and complexity aspects of logics in team semantics have been extensively studied (see, \eg, \cite{Durand:2022:tractability}) and  interesting connections have been found to areas such as database theory \cite{Kontinen:2013:independence,Hannula:2020:polyteam}, Bayesian networks \cite{Corander:2019}, quantum foundations \cite{Abramsky:2021:team}, and inquisitive and separation logic \cite{Ciardelli:2020,Haase:2022}. These works have focused on logics in the first-order, propositional and modal team semantics, and more recently also in the multi-set and probabilistic settings \cite{Abramsky:2021:team,Durand:2018:approximation,Durand:2018:probabilistic}. 

A defining feature of team semantics is that satisfaction of a formula is defined with reference to a \emphexpr{set} of assignments, not just a \emphexpr{single} assignment, as is the case in classical first-order logic. In this respect, logics with team semantics resemble modal logic, where a Kripke structure offers a spectrum of possible assignments. In a Kripke structure, there is the so-called \partexpr{accessibility relation} which brings in the aspect that assignments, for example, develop in time. This makes it possible to talk in modal logic  about \emphexpr{necessity} and \emphexpr{possibility}. Team semantics is more rigid. There is just the \emphexpr{set} (called \announceexpr{team}) of assignments without any built-in relations between its elements. However, this makes it possible to express interesting and highly non-trivial \quoexpr{combinatorial} properties of variables, such as the aforementioned functional dependence but also inclusion and exclusion atoms \cite{Galliani:2012} and independence atoms \cite{Gradel:2013}. 

We can think of the sub-team relation as an accessibility relation, and then \quoexpr{necessity} and \quoexpr{possibility} make perfect sense. Then \extexpr{necessarily $\phi$} means \extexpr{$\phi$ is true in every sub-team}, and \extexpr{possibly $\phi$} means \extexpr{$\phi$ is true in some non-empty sub-team}.

Some formulas do not seem to take advantage of the existence of a team around an assignment. One way to make this phenomenon exact is the following: a formula $\phi$ is called \defexpr{flat}, or is said to satisfy the \defexpr{flatness criterion}, if for all model $\mathfrak{M}$ and for all team $X$ the following holds
\begin{equation}\label{eqn:flat}
	\mathfrak{M}\vDashteam_{X}\phi\iff\forall s\in X\colon\mathfrak{M}\vDashteam_{\{s\}}\phi\,.
\end{equation}
This criterion holds for first-order formulas $\phi$, \ie{} all first-order formulas are flat. This offers a method to show that some given formula is not logically equivalent to a first-order formula. One just shows that the formula is not flat. To take two simple examples, $\dep(x)$ and $\dep(x,y)$ are not flat, hence they cannot be logically equivalent to a first-order formula, which is intuitively obvious but still needs an argument. Flatness provides such an argument. In general, the problem of deciding whether a given dependence logic formula is equivalent  to a first-order formula is highly undecidable. Therefore, the question of first-order expressibility is non-trivial and any method, even a partial one, is useful. 

One can say that flat formulas are, in a sense, \quoexpr{local}, in that their truth is determined by what holds for single assignments. Still, non-trivial properties of teams can be expressed. For example, $x\syntacticeq y$ is flat and it says that the variables $x$ and $y$ have the same value in the team. It is a trivial property of a team but still such teams are clearly distinguished from totally arbitrary teams. Flatness just means that the distinguishing property is \quoexpr{flatly spread out} in the team.

One might think that flatness of a formula is reflected in the formula being built from flat elements. But all \emphexpr{sentences} are flat because their truth in any non-empty team is determined by their truth in $\{\varnothing\}$, \ie{} the team containing only the empty assignment. A sentence expresses something about the model that the team is based on. A sentence says nothing about any team. The flatness of sentences is, therefore, like an anomaly. But it follows that being flat is not equivalent to being, up to logical equivalence, first-order. A flat formula may very well have non-flat sub-formulas. 

\begin{example*}
	We consider a signature with a $2$-ary relation $R(\emptyarg,\emptyarg)$. Consider the following formula $\phi(x)$
	\[
		\exists u(R(x,u)\land\forall v(R(x,v)\rightarrow\exists w(R(x,w)\land \dep(w,v)\land w\syntacticneq u)))\,.
	\]
	This is satisfied by a team $X$ if for every assignment of the variable $x$ in the team $X$, the set $R(x,\emptyarg)$ is infinite. It is a flat formula, but its sub-formula $\dep(w,v)$ is not flat. Thus, condition (\ref{eqn:flat}) only says that $\phi$ itself behaves in a first order way by being flat, but its sub-formulas need not do so. 
\end{example*}

The class of teams satisfying a flat formula is closed downwards and closed under unions and intersections. There is a maximal team, and the class of teams satisfying a flat formula is precisely the power-set of that team. The maximal team of a flat formula $\phi$ is given by
\[
	\left\{s\setsep\mathfrak{M}\vDashteam_{\{s\}}\phi\right\}\,.
\]
Thus the dimension (in the sense of \cite{Hella:2024}) of a flat formula is always $1$. 

We can make an arbitrary formula flat by using the following new logical operator, first introduced by Hodges in 1997 \cite{Hodges:1997:some}, called the \defexpr{flattening operator} and denoted here by $\F$:
\begin{equation}\label{eqn:flattening-operator}
	\mathfrak{M}\vDashteam_{X}\F\phi\iff\forall s\in X\colon\mathfrak{M}\vDashteam_{\{s\}}\phi\,.
\end{equation}
Thus, $\F\phi$ is always flat, and for a flat formula $\phi$, we have $\phi\equiv\F\phi$. The flattening operator is somewhat similar to the necessity operator $\Box\phi$ and the possibility operator $\Diamond\phi$ from Modal Logic. The \quoexpr{mode} of $\phi$ that $\F\phi$ brings about is not that $\phi$ is \quoexpr{necessary} or \quoexpr{possible}, but that on the question whether the entire team satisfies $\phi$ or not, at least we know that $\phi$ is true in singleton sub-teams. We can perhaps anticipate that $\F\phi$ is computationally much simpler to check for satisfaction than $\phi$ itself, since we have to check only single assignments. However, if $\phi$ is a sentence, then $\F\phi\equiv\phi$. So $\F\phi$ may be computationally as complicated as $\phi$ itself.

The flattening operator $\F$ makes an arbitrary formula flat. This happens with no regard to what the formula looks like inside. There is another way to make a formula flat. This is based on changing the formula inside as follows: The \emphexpr{flattening} $\phi^{\f}$ of a formula $\phi$ is obtained inductively by replacing non-first-order atomic formulas by something that is flat. If there are non-first-order logical operations, then they are similarly replaced by something flat.
We can impose axioms for the flattening, because it is \apriori{} not clear what it is. There may be different ways to flatten a formula, but we expect them all to satisfy some basic axioms. This will be discussed in \cref{sec:flattening-of-a-formula}.

Although $\F\phi$ is always flat, it is not always equivalent to $\phi^{\f}$. We can use the flattening of a formula to define a kind of \emphexpr{negation} of a formula: we take the flattening $\phi^{\f}$, which is first-order, and then apply negation to the flattening. 

A weaker version of flatness is the following: a formula is \defexpr{downwards flat} if 
\begin{equation}\label{eqn:downwards-flat}
	\mathfrak{M}\vDashteam_{X}\phi\implies\forall s\in X\colon\mathfrak{M}\vDashteam_{\{s\}}\phi\,.
\end{equation}
Thus, the truth of $\phi$ in a team is inherited by the singleton sub-teams. Of course, downwards closed formulas are downwards flat, but not necessarily flat. An example is $\dep(x)$. Every singleton team satisfies $\dep(x)$ but not every team satisfies $\dep(x)$.

A similar weaker version of flatness is the following: a formula is \defexpr{upwards flat} if 
\begin{equation}\label{eqn:upwards-flat}
	(\forall s\in X\colon\mathfrak{M}\vDashteam_{\{s\}}\phi)\implies\mathfrak{M}\vDashteam_{X}\phi\,.
\end{equation}
So for the truth of $\phi$ in a team, it is enough that $\phi$ \quoexpr{gathers} truth from little singleton pieces inside the team.

Still another variation of flatness is \emph{$n$-coherence} (introduced in \cite{Kontinen:2013:coherence}, see also \cite{Meissner:2022} and \cite{Ciardelli:2022}):
\begin{equation}\label{eqn:n-flat}
	\mathfrak{M}\vDashteam_{X}\phi\iff(\forall Y\subseteq X\text{ such that }\left|Y\right|=n,\ \mathfrak{M}\vDashteam_{Y}\phi)\,.
\end{equation}
Usually, flatness corresponds in this terminology to $1$-coherence. For example, dependence atoms are not flat but they are $2$-coherent. In \cite{Kontinen:2013:coherence}, it is shown that $n$-coherence is preserved by conjunction but not by disjunction, except if one of the disjuncts is $1$-coherent. The disjunction of two dependence atoms need not be $n$-coherent for any $n$. The model-checking problem for $n$-coherent formulas is in $\LOGSPACE$. Model-checking of the disjunction of two $2$-coherent formulas is in $\NLOGSPACE$. In \cite{Kontinen:2013:coherence}, it is asked whether there is a dependence logic formula which is not $k$-coherent for any $k$ but which is \quoexpr{$\sqrt{n}$-coherent}?

Coherence, just like flatness, can be approached from the point of view of an operator as well. The \emphexpr{$n$-coherence operator} applied to a formula says of a team that every sub-team of size $n$ satisfies the formula. The \emphexpr{co-$n$-coherence operator} applied to a formula says of a team of size $m$ that every sub-team of size $m-n$ satisfies the formula. The co-$n$-coherence operator is a way to say that a formula is approximately satisfied by the team. In \cite{Vaananen:2017} the approximation is more generous (proportional). We can even let an operator say that the team satisfies the formula even if any $n$ assignments satisfying a given fixed first order formula are \emphexpr{added} to the team or \emphexpr{taken away} from the team. So, the team satisfies the formula even if it is made a little smaller or a little bigger. It is a kind of \quoexpr{give or take at most $n$} operator. Dependence logic is closed under all these operators. We will not pursue this line of argument further in this paper.

\paragraph{Our contributions}

We study the effects of extending several team-based logics by the flattening operator of Hodges. We also take a new look at the concept of a flattening of a formula that has been utilized in several works in the area. In \cref{sec:preliminaries} we recall the basic concepts and definitions relevant for this work. In \cref{sec:flattening-of-a-formula} we show that the flattening of a formula can be characterized axiomatically but also give explicit inductive definitions of it for the most prominent team-based logics. In \cref{sec:flattening-operator} we introduce the flattening operator and study its interactions with the aforementioned logics. In \cref{sec:expressive-power-of-F} we show that the flattening operator increases the expressive power of the unary fragments of inclusion and anonymity logics. These results show that the two different ways of defining a flattening of a formula are not in general equivalent.

\section{Preliminaries}
\label{sec:preliminaries}

In this section, we will briefly recall the notation used in this work, the definition of team semantics, and some basic results that will be used in the rest of this work. Through all of this work, we assume that we have countable sets of individual variable symbols $x_{i}$, $y_{i}$, $z_{i}$, etc. for $i\in\mathbb{N}$, and of relation variable symbols $X_{i}$, $Y_{i}$, $Z_{i}$, etc. for $i\in\mathbb{N}$ of all arities. We will use Fraktur font to represent models, such as $\mathfrak{A}$, $\mathfrak{B}$, $\mathfrak{M}$, etc. and use Roman font to represent their corresponding domain $A$, $B$, $M$, etc. We will write $\tuple{x}$, $\tuple{y}$, $\tuple{z}$ and so on to describe tuples of variable symbols; and likewise, we will write $\tuple{a}$, $\tuple{b}$, $\tuple{m}$ and so forth to describe tuples of elements of a model. Given any set $A$, we will furthermore write $\powerset(A)$ for the powerset $\left\{B\setsep B \subseteq A\right\}$ of $A$, and $\powersetstar(A)$ for $\powerset(A)\setminus\left\{\varnothing\right\}$.

\subsection{Team semantics}

If $s$ is an assignment, we write $s\left[m\replaces v\right]$ for the assignment in which we substituted the variable $v$ by the value $m$. We recall from \cite{Vaananen:2007:dependence} the following concepts: a \defexpr{team} $X$ is any set of assignments for a fixed set of variables, denoted $\dom(X)$. Suppose now $X$ is a team. It is sometimes useful to think of $X$ as a \emph{relation} in the following sense: suppose $x_1,\ldots,x_n\in\dom(X)$. We obtain from $X$ the $n$-ary relation $\left\{\langle s(x_1),\ldots,s(x_n)\rangle\setsep s\in X\right\}$. If $m$ is an element the domain $M$ of our model $\mathfrak{M}$, then $X[m/v]$ is the team consisting of all $s[m/v]$, where $s\in X$. We use $X[M/v]$, called the \defexpr{duplication} of $X$ at $v$, to denote the set of all $s[m/v]$, where $m\in M$  and $s\in X$. Finally, if $H\colon X\to\powersetstar(M)$, then $X[H/v]$, called the \defexpr{supplementation} of $X$ at $v$ by $H$, is the team consisting of all $s[m/v]$, where $s\in X$ and $m\in H(s)$.

For our purposes, it will be useful to first present team semantics for first-order logic proper, and then progressively add new atoms (dependency, non-dependency, inclusion, exclusion, etc.).

We recall the usual definition of team semantics (in its lax formulation). For the sake of simplicity, it is assumed that all expressions are in \partexpr{negation normal form}, meaning that negations $\lnot$ only appear in front of first-order literals.

\begin{definition}[Team semantics for first-order logic]\label{def:team-semantics-FO}
	Let $\sigma$ be a first-order signature. Let $\mathfrak{M}$ be a first-order $\sigma$-model, let $\varphi$ be a first-order $\sigma$-formula, and let $X$ be a team over $\mathfrak{M}$ such that $\dom\left(X\right)\supseteq\FV\left(\varphi\right)$. Then we define the relation $\mathfrak{M}\vDashteam_{X}\varphi$ by induction over the structure of the formula $\varphi$ as follows:
	\begin{description}
		\item[TS-lit:]\label{def:TS-lit} If $\varphi\equiv\alpha$ where $\alpha$ is a first-order $\sigma$-literal, then $\mathfrak{M}\vDashteam_{X}\varphi$ if and only if for all assignments $s\in X$ it holds that $\mathfrak{M}\vDashTarski_{s}\alpha$ (here according to Tarskian semantics).
		
		\item[TS-$\lor$:]\label{def:TS-or} If $\varphi\equiv\psi_{1}\lor\psi_{2}$ where $\psi_{1}$ and $\psi_{2}$ are first-order $\sigma$-formulas, then $\mathfrak{M}\vDashteam_{X}\varphi$ if and only if there exist some teams $Y$ and $Z$ (not necessarily disjoint) such that $X=Y\cup Z$, $\mathfrak{M}\vDashteam_{Y}\psi_{1}$ and $\mathfrak{M}\vDashteam_{Z}\psi_{2}$.
		
		\item[TS-$\land$:]\label{def:TS-and} If $\varphi\equiv\psi_{1}\land\psi_{2}$ where $\psi_{1}$ and $\psi_{2}$ are first-order $\sigma$-formulas, then $\mathfrak{M}\vDashteam_{X}\varphi$ if and only if $\mathfrak{M}\vDashteam_{X}\psi_{1}$ and $\mathfrak{M}\vDashteam_{X}\psi_{2}$.
		
		\item[TS-$\exists$:]\label{def:TS-exists} If $\varphi\equiv\exists v\psi$ where $v$ is a variable symbol and $\psi$ is a first-order $\sigma$-formula, then $\mathfrak{M}\vDashteam_{X}\varphi$ if and only if there exists some $H\colon X\to\powersetstar\left(M\right)$ such that $\mathfrak{M}\vDashteam_{X\left[H\replaces v\right]}\psi$.
		
		\item[TS-$\forall$:]\label{def:TS-forall} If $\varphi\equiv\forall v\psi$ where $v$ is a variable symbol and $\psi$ is a first-order $\sigma$-formula, then $\mathfrak{M}\vDashteam_{X}\varphi$ if and only if $\mathfrak{M}\vDashteam_{X\left[M\replaces v\right]}\psi$.
	\end{description}
\end{definition}

If it is the case that $\mathfrak{M}\vDashteam_{X}\varphi$, then we say that the model $\mathfrak{M}$ together with the team $X$ \defexpr{satisfy} $\varphi$, and if it is not the case, then we say that the model $\mathfrak{M}$ together with the team $X$ \defexpr{do not satisfy} $\varphi$ and we write it $\mathfrak{M}\nvDashteam_{X}\varphi$. If $\varphi$ is a sentence (\ie{} has no free-variables), we say that $\varphi$ is \defexpr{true} in $\mathfrak{M}$ according to team semantics, and we write it $\mathfrak{M}\vDashteam\varphi$, if $\mathfrak{M}\vDashteam_{\left\{\varnothing\right\}}\varphi$ where $\left\{\varnothing\right\}$ is the team containing only the empty assignment. Otherwise, we say that $\varphi$ is \defexpr{false} in $\mathfrak{M}$ according to team semantics, and we write it $\mathfrak{M}\nvDashteam\varphi$.

We recall the most well-known team-based logics. We first give definitions for \announceexpr{dependence atoms}, introduced in \cite{Vaananen:2007:dependence}. Let $k\in\mathbb{N}$. If $\tuple{t}_{1}$ is a $k$-tuple of $\sigma$-terms and $t_{2}$ is a $\sigma$-term, then $\dep(\tuple{t}_{1},t_{2})$ is called a \defexpr{$k$-ary dependence atom}. For $0$-ary atoms, called \defexpr{constancy atoms}, we may simply write $\dep(t_{2})$. We define $\FV(\dep(\tuple{t}_{1},t_{2}))=\Vr(\tuple{t}_{1})\cup\Vr(t_{2})$. The set of formulas of \defexpr{dependence logic} over a signature $\sigma$, denoted by $\FO(\depformat)\left[\sigma\right]$ or even more simply $\DEP\left[\sigma\right]$, is defined by adding dependence atoms of all arities to the definition of the set of formulas of $\FO\left[\sigma\right]$. By allowing only at most $k$-ary dependence atoms in the condition above, we obtain the fragment $\FO(\depformat[k])\left[\sigma\right]$ or even more simply $\DEP[k]\left[\sigma\right]$ called \defexpr{$k$-ary dependence logic}. If $\mathfrak{M}$ is a model and $X$ is a team over $\mathfrak{M}$ such that $\dom(X)\supseteq\Vr(\tuple{t}_{1}t_{2})$, then we define the truth of $\dep(\tuple{t}_{1},t_{2})$ in the model $\mathfrak{M}$ and the team $X$:
\begin{center}
	$\mathfrak{M}\vDashteam_{X}\dep(\tuple{t}_{1},t_{2})$\quad iff\quad for all $s,s'\in X$, if $s(\tuple{t}_{1})=s'(\tuple{t}_{1})$, then $s(t_{2})=s'(t_{2})$.
\end{center}

Next, we present so-called, \announceexpr{anonymity atoms} also called \announceexpr{non-dependency atoms} (\cite{Galliani:2012}), or \announceexpr{afunctional dependence atoms} (\cite{Paredaens:1989}). They have been studied, \eg, in \cite{Vaananen:2022,Ronnholm:2018:phdthesis}. Let $k\in\mathbb{N}$. If $\tuple{t}_{1}$ is a $k$-tuple of $\sigma$-terms and $t_{2}$ is a $\sigma$-term, then $\tuple{t}_{1}\anon t_{2}$ is called a \defexpr{$k$-ary anonymity atom}. For $0$-ary atoms $\anon[0]t$, called \defexpr{non-constancy atoms}, we may simply write $\anon t$. We define $\FV(\tuple{t}_{1}\anon t_{2})=\Vr(\tuple{t}_{1})\cup\Vr(t_{2})$. The set of formulas of \defexpr{anonymity logic} over a signature $\sigma$, denoted by $\FO(\anon)\left[\sigma\right]$ or even more simply $\ANON\left[\sigma\right]$, is defined by adding anonymity atoms of all arities to the definition of the set of formulas of $\FO\left[\sigma\right]$. By allowing only at most $k$-ary anonymity atoms in the condition above, we obtain the fragment $\FO(\anon[k])\left[\sigma\right]$ or even more simply $\ANON[k]\left[\sigma\right]$ called \defexpr{$k$-ary anonymity logic}. If $\mathfrak{M}$ is a model and $X$ is a team over $\mathfrak{M}$ such that $\dom(X)\supseteq\Vr(\tuple{t}_{1}t_{2})$, then we define the truth of $\tuple{t}_{1}\anon t_{2}$ in the model $\mathfrak{M}$ and the team $X$:
\begin{center}
	$\mathfrak{M}\vDashteam_{X}\tuple{t}_{1}\anon t_{2}$\quad iff\quad for every $s\in X$, there is $s'\in X$ such that $s(\tuple{t}_{1})=s'(\tuple{t}_{1})$ but $s(t_{2})\neq s'(t_{2})$.
\end{center}

Next, we turn to \announceexpr{inclusion atoms}. Let $k\in\mathbb{N}$. If $\tuple{t}_{1},\tuple{t}_{2}$ are $k$-tuples of $\sigma$-terms, then $\tuple{t}_{1}\incl\tuple{t}_{2}$ is a \defexpr{$k$-ary inclusion atom}. We define $\FV(\tuple{t}_{1}\incl\tuple{t}_{2})=\Vr(\tuple{t}_{1})\cup\Vr(\tuple{t}_{2})$. The set of formulas of \defexpr{inclusion logic} over a signature $\sigma$, denoted by $\FO(\incl)\left[\sigma\right]$ or even more simply $\INC\left[\sigma\right]$, is defined by adding inclusion atoms of all arities to the definition of the set of formulas of $\FO\left[\sigma\right]$. By allowing only at most $k$-ary inclusion atoms in the condition above, we obtain the fragment $\FO(\incl[k])\left[\sigma\right]$ or even more simply $\INC[k]\left[\sigma\right]$ called \defexpr{$k$-ary inclusion logic}. If $\mathfrak{M}$ is a model and $X$ is a team over $\mathfrak{M}$ such that $\dom(X)\supseteq\Vr(\tuple{t}_{1}\tuple{t}_{2})$, then we define the truth of $\tuple{t}_{1}\incl\tuple{t}_{2}$ in the model $\mathfrak{M}$ and the team $X$:
\begin{center}
	$\mathfrak{M}\vDashteam_{X}\tuple{t}_{1}\incl\tuple{t}_{2}$\quad iff\quad for all $s\in X$, there exists $s'\in X$ such that $s(\tuple{t}_{1})=s'(\tuple{t}_{2})$.
\end{center}

Let us now recall \announceexpr{(conditional) independence atoms}. Let $k_{1},k_{2},k_{3}\in\mathbb{N}$. If $\tuple{t}_{1}$ is a $k_{1}$-tuple, $\tuple{t}_{2}$ a $k_{2}$-tuple of $\sigma$-terms and $\tuple{t}_{3}$ a $k_{3}$-tuple of $\sigma$-terms, then $\tuple{t}_{2}\ind_{\tuple{t}_{1}}\tuple{t}_{3}$ is a \defexpr{($k_{1}$,$k_{2},k_{3}$)-ary (conditional) independence atom}. We define $\FV(\tuple{t}_{2}\ind_{\tuple{t}_{1}}\tuple{t}_{3})=\Vr(\tuple{t}_{1})\cup\Vr(\tuple{t}_{2})\cup\Vr(\tuple{t}_{3})$. The set of formulas of \defexpr{(conditional) independence logic} over a signature $\sigma$, denoted by $\FOind\left[\sigma\right]$ or even more simply $\IND\left[\sigma\right]$ is defined by adding independence atoms of all arities to the definition of the set of formulas of $\FO\left[\sigma\right]$. By allowing only at most ($k_{1}$,$k_{2},k_{3}$)-ary independence atoms in the condition above, we obtain the fragment $\FO(\ind_{k_{1},k_{2},k_{3}})\left[\sigma\right]$ or even more simply $\IND[k_{1},k_{2},k_{3}]\left[\sigma\right]$ called \defexpr{$(k_{1},k_{2},k_{3})$-ary (conditional) independence logic}. If $\mathfrak{M}$ is a model and $X$ is a team over $\mathfrak{M}$ such that $\dom(X)\supseteq\Vr(\tuple{t}_{1}\tuple{t}_{2}\tuple{t}_{3})$, then we define the truth of $\tuple{t}_{2}\ind_{\tuple{t}_{1}}\tuple{t}_{3}$ in the model $\mathfrak{M}$ and the team $X$:
\begin{center}
	$\mathfrak{M}\vDashteam_{X}\tuple{t}_{2}\ind_{\tuple{t}_{1}}\tuple{t}_{3}$\quad iff\quad for all $s,s'\in X$, if $s(\tuple{t}_{1})=s'(\tuple{t}_{1})$, then there exists $s''\in X$ such that $s''(\tuple{t}_{1}\tuple{t}_{2})=s(\tuple{t}_{1}\tuple{t}_{2})$ and $s''(\tuple{t}_{1}\tuple{t}_{3})=s'(\tuple{t}_{1}\tuple{t}_{3})$.
\end{center}

Finally, we present the syntax and the semantics for \announceexpr{exclusion logic}. Let $k\in\mathbb{N}$. If $\tuple{t}_{1},\tuple{t}_{2}$ are $k$-tuples of $\sigma$-terms, then $\tuple{t}_{1}\excl\tuple{t}_{2}$ is a \defexpr{$k$-ary exclusion atom}. We define $\FV(\tuple{t}_{1}\excl\tuple{t}_{2})=\Vr(\tuple{t}_{1})\cup\Vr(\tuple{t}_{2})$. The set of formulas of \defexpr{exclusion logic} over a signature $\sigma$, denoted by $\FO(\excl)\left[\sigma\right]$ or even more simply $\EXC\left[\sigma\right]$, is defined by adding exclusion atoms of all arities to the definition of the set of formulas of $\FO\left[\sigma\right]$. By allowing only at most $k$-ary exclusion atoms in the condition above, we obtain the fragment $\FO(\excl[k])\left[\sigma\right]$ or even more simply $\EXC[k]\left[\sigma\right]$ called \defexpr{$k$-ary exclusion logic}. If $\mathfrak{M}$ is a model and $X$ is a team over $\mathfrak{M}$ such that $\dom(X)\supseteq\Vr(\tuple{t}_{1}\tuple{t}_{2})$. We define the truth of $\tuple{t}_{1}\excl\tuple{t}_{2}$ in the model $\mathfrak{M}$ and the team $X$:
\begin{center}
	$\mathfrak{M}\vDashteam_{X}\tuple{t}_{1}\excl\tuple{t}_{2}$\quad iff\quad for all $s,s'\in X$, it holds that $s(\tuple{t}_{1})\neq s'(\tuple{t}_{2})$.
\end{center}

More generally, we will speak of \defexpr{team-based logic} to speak of one of these logics based on team semantics (\ie{} independence logic, anonymity logic, inclusion logic, exclusion logic, or independence logic) and of \defexpr{team-based formula} to speak of a formula from one of these logics (\ie{} a formula which uses dependency atoms, anonymity atoms, inclusion atoms, exclusion atoms, or independence atoms).

\subsection{Closure properties}

The logics that we have presented so far can be differentiated on the basis of their distinct closure properties. Let us start by defining the closure properties.

\begin{definition}
	Let $\sigma$ be a signature. A team-based $\sigma$-formula $\varphi$:
	\begin{itemize}
		\item is said to have the \defexpr{empty team property} if for any $\sigma$-model $\mathfrak{M}$, it holds that $\mathfrak{M}\vDashteam_{\varnothing}\varphi$ (according to team semantics).
		
		\item is said to be \defexpr{downwards closed} if for all $\sigma$-models $\mathfrak{M}$ and teams $X$ over $\mathfrak{M}$, it holds that
		\[
			\mathfrak{M}\vDashteam_{X}\varphi\implies\forall Y\subseteq X,\mathfrak{M}\vDashteam_{Y}\varphi\,.
		\]
		
		\item is said to be \defexpr{union closed} if for all $\sigma$-models $\mathfrak{M}$ and collections $\left\{X_i\setsep i\in I\right\}$ of teams $X_i$ over $\mathfrak{M}$, it holds that
		\[
			\mathfrak{M}\vDashteam_{X_i}\varphi\text{ for all $i\in I$}\implies\mathfrak{M}\vDashteam_{\cup_{i\in I}X_i}\varphi\,.
		\]
		
		\item is said to be \defexpr{flat} if for all $\sigma$-models $\mathfrak{M}$ and teams $X$ over $\mathfrak{M}$, it holds that
		\[
			\mathfrak{M}\vDashteam_{X}\varphi\iff\mathfrak{M}\vDashteam_{\left\{s\right\}}\varphi\text{ for all }s\in X\,.
		\]
	\end{itemize}
\end{definition}

One can remark that the flatness property is equivalent to the combination of the downwards closure and the union closure properties.

Each of these properties is extended to team-based logics: a given logic (\eg{} $\DEP$, $\ANON$, $\INC$, $\EXC$, \dots) has some of these properties if \emphexpr{all} formulas of the logic have  those properties. We can then split some of the well-known team-based logics between the ones that are downwards closed and the ones that are union closed.

\section{Flattening of a formula}
\label{sec:flattening-of-a-formula}

We will now introduce the \announceexpr{flattening of a formula} $\varphi$, denoted by $\varphi^{\f}$, where $\varphi$ belongs to some team-based logic. The notion of \extexpr{flattened formulas} was first defined  in \cite{Vaananen:2007:dependence} and then used \eg{} in \cite{Yang:2014}.

The definition of flattening that was given in \cite{Vaananen:2007:dependence} was inductive and for dependence logic only. We present here an axiomatic definition which is appropriate in a wider context. We show that in most cases these axioms determine a unique inductive definition of a flattening of a formula.   %
\begin{enumerate}
	\item\label{ax:entailment-axiom} (entailment axiom) for all models $\mathfrak{M}$ and all teams $X$ over $\mathfrak{M}$ such that $\dom(X)\supseteq\FV(\varphi)$, the following must hold
	\[
		\mathfrak{M}\vDash_{X}\varphi\implies\mathfrak{M}\vDash_{X}\varphi^{\f}\,.
	\]
	
	\item\label{ax:first-order-axiom} (flatness axiom) $\varphi^{\f}$ is a flat formula.
	
	\item\label{ax:distributivity-axiom} (distributivity axiom) The following holds: $(\star(\varphi_{1},\dots,\varphi_{n}))^{\f}=\star(\varphi_{1}^{\f},\dots,\varphi_{n}^{\f})$ for any logical operator $\star$ of arity $n$, and $(Qx\varphi)^{\f}=Qx\varphi^{\f}$ for any quantifier $Q$.
\end{enumerate}

Now, what we would like to obtain is an explicit inductive definition of the flattening of a given formula. From these requirements, we can infer the flattening of the most well-known team-based atoms. It is easy to see that taking $\top$ for the flattening of any non first-order atom is \emphexpr{sufficient} to satisfy all these requirements. But the interesting question is: is it \emphexpr{necessary}? The goal of this section is to answer this question.

We get the following flattening for the most well-known team-based atoms

\begin{proposition}\label{prp:flattening-atoms}
The following equivalences are implied by the above axioms:
	\begin{multicols}{2}
		\begin{enumerate}
			\item\label{itm:flattening-atoms:dep} for dependency atoms: $(\dep(\tuple{x},y))^{\f}\equiv\top$,
			
			\item\label{itm:flattening-atoms:anon} for anonymity atoms: $(\tuple{x}\anon y)^{\f}\equiv\top$,
			
			\item\label{itm:flattening-atoms:incl} for inclusion atoms: $(\tuple{x}\incl\tuple{y})^{\f}\equiv\top$,
			
			\item\label{itm:flattening-atoms:ind} for independence atoms: $(\tuple{x}\ind\tuple{y})^{\f}\equiv\top$.
		\end{enumerate}
	\end{multicols}
\end{proposition}

\begin{proof}
	The proof of each claim is straightforward.
	\begin{enumerate}
		\item[\ref{itm:flattening-atoms:dep}] Suppose  $(\dep(\tuple{x},y))^{\f}\nequiv\top$. Then $\mathfrak{M}\nvDash_{X}(\dep(\tuple{x},y))^{\f}$ for some model $\mathfrak{M}$ and team $X$. But by flatness axiom, $(\dep(\tuple{x},y))^{\f}$ is flat, and so it means that there is $s_{0}\in X$ such that $\mathfrak{M}\nvDash_{\left\{s_{0}\right\}}(\dep(\tuple{x},y))^{\f}$.
		On the other hand, as a singleton team satisfies any  dependence atom, then in particular $\mathfrak{M}\vDash_{\left\{s_{0}\right\}}\dep(\tuple{x},y)$, and so by the entailment axiom, $\mathfrak{M}\vDash_{\left\{s_{0}\right\}}(\dep(\tuple{x},y))^{\f}$. Contradiction.
		
		\item[\ref{itm:flattening-atoms:anon}] Suppose that $(\tuple{x}\anon y)^{\f}\nequiv\top$. Then $\mathfrak{M}\nvDash_{X}(\tuple{x}\anon y)^{\f}$ for some model $\mathfrak{M}$ and team $X$. But by flatness axiom, $(\tuple{x}\anon y)^{\f}$ is flat, and so it means that there is $s_{0}\in X$ such that $\mathfrak{M}\nvDash_{\left\{s_{0}\right\}}(\tuple{x}\anon y)^{\f}$.
		On the other hand, we construct from scratch a new team $Y$ as follows. We write $\tuple{a}=s_{0}(\tuple{x})$ and $b=s_{0}(y)$. Let $b'\in M\setminus\left\{b\right\}$ (non-empty by assumption in this article). Let $s_{0}'$ be the assignment $\left\{\tuple{x}\mapsto\tuple{a},y\mapsto b'\right\}$. Then consider the team $Y=\left\{s_{0},s_{0}'\right\}$, \ie{} $Y$ is the following team
		\begin{center}
			\begin{tabular}{cc|c}
				& $\tuple{x}$ & $y$ \\\cline{2-3}
				$s_{0}\colon$  & $\tuple{a}$ & $b$ \\
				$s_{0}'\colon$ & $\tuple{a}$ & $b'$
			\end{tabular}
		\end{center}
		By construction, $\mathfrak{M}\vDash_{Y}\tuple{x}\anon y$, so by the entailment axiom, we get $\mathfrak{M}\vDash_{Y}(\tuple{x}\anon y)^{\f}$. Now by the flatness axiom, $(\tuple{x}\anon y)^{\f}$ is flat, and so for all $s\in Y$ we get $\mathfrak{M}\vDash_{\left\{s\right\}}(\tuple{x}\anon y)^{\f}$, in particular for $s_{0}\in Y$ we get $\mathfrak{M}\vDash_{\left\{s_{0}\right\}}(\tuple{x}\anon y)^{\f}$. Contradiction.
		
		\item[\ref{itm:flattening-atoms:incl}] Suppose  $(\tuple{x}\incl\tuple{y})^{\f}\nequiv\top$. Then $\mathfrak{M}\nvDash_{X}(\tuple{x}\incl\tuple{y})^{\f}$ for some model $\mathfrak{M}$ and team $X$. But by flatness axiom, $(\tuple{x}\incl\tuple{y})^{\f}$ is flat, and so it means that there is $s_{0}\in X$ such that $\mathfrak{M}\nvDash_{\left\{s_{0}\right\}}(\tuple{x}\incl\tuple{y})^{\f}$.
		On the other hand, we construct from scratch a new team $Y$ as follows. We write $\tuple{a}=s_{0}(\tuple{x})$ and $\tuple{b}=s_{0}(\tuple{y})$. Let $s_{0}'$ be the assignment $\left\{\tuple{x}\mapsto\tuple{a},\tuple{y}\mapsto\tuple{a}\right\}$. Then consider the team $Y=\left\{s_{0},s_{0}'\right\}$, \ie{} $Y$ is the following team
		\begin{center}
			\begin{tabular}{cc|c}
				& $\tuple{x}$ & $\tuple{y}$\\\cline{2-3}
				$s_{0}\colon$  & $\tuple{a}$ & $\tuple{b}$\\
				$s_{0}'\colon$ & $\tuple{a}$ & $\tuple{a}$
			\end{tabular}
		\end{center}
		By construction, $\mathfrak{M}\vDash_{Y}\tuple{x}\incl\tuple{y}$, so by the entailment axiom, we get $\mathfrak{M}\vDash_{Y}(\tuple{x}\incl\tuple{y})^{\f}$. Now by the flatness axiom, $(\tuple{x}\incl\tuple{y})^{\f}$ is flat, and so for all $s\in Y$ we get $\mathfrak{M}\vDash_{\left\{s\right\}}(\tuple{x}\incl\tuple{y})^{\f}$, in particular for $s_{0}\in Y$ we get $\mathfrak{M}\vDash_{\left\{s_{0}\right\}}(\tuple{x}\incl\tuple{y})^{\f}$. Contradiction.
		
		\item[\ref{itm:flattening-atoms:ind}] Suppose $(\tuple{x}\ind\tuple{y})^{\f}\nequiv\top$. Then $\mathfrak{M}\nvDash_{X}(\tuple{x}\ind\tuple{y})^{\f}$ for some model $\mathfrak{M}$ and team $X$. But by flatness axiom, $(\tuple{x}\ind\tuple{y})^{\f}$ is flat, and so it means that there is $s_{0}\in X$ such that $\mathfrak{M}\nvDash_{\left\{s_{0}\right\}}(\tuple{x}\ind\tuple{y})^{\f}$.
		On the other hand, as a singleton team satisfy any pure independence atom, then in particular $\mathfrak{M}\vDash_{\left\{s_{0}\right\}}\tuple{x}\ind\tuple{y}$, and so by the entailment axiom, $\mathfrak{M}\vDash_{\left\{s_{0}\right\}}(\tuple{x}\ind\tuple{y})^{\f}$. Contradiction.
	\end{enumerate}
\end{proof}

\begin{remark}
	Curiously, we do not obtain a unique solution in the case of the exclusion atom. Indeed, for the exclusion atom $(\tuple{x}\excl\tuple{y})^{\f}$, there are (at least) two choices satisfying the flattening axioms: we can let $(\tuple{x}\excl\tuple{y})^{\f}$ be $\top$, or can let $(\tuple{x}\excl\tuple{y})^{\f}$ be $\tuple{x}\neq\tuple{y}$. Both choices  satisfy the flattening axioms.

	In order to remain as symmetrical as possible, we make the choice in this thesis to take $\top$ as the flattening of the exclusion atom.
\end{remark}

Note that in this \cref{prp:flattening-atoms}, we found the flattening only for atoms built of variables, not terms. In fact, this is not an issue as we can eliminate composite terms inside team-based atoms by existentially quantifying over new variables. Hence, we now have an explici expression for all team-based atoms built of terms. Also, it is immediate to infer that the flattening of a first-order literal in actually itself. One can then infer the flattening of any team-based formula using the \hyperref[ax:distributivity-axiom]{distributivity axiom}: substitute all team-based atoms by $\top$, and keep all other (first-order) literals as they are. %

\section{Flattening operator}
\label{sec:flattening-operator}

We move now to investigating the flattening operator which just boldly asserts that the formula following the operator has to be treated in a flat way, whether the formula looks flat or not. This operator was first introduced  in \cite{Hodges:1997:some} and was at that time denoted by \citeexpr{$\downarrow$}. In this article, we will denote it by \citeexpr{$\F$}.

\subsection{Definition}

First we give the syntax and the semantics of the flattening operator $\F$.

\begin{definition}[Syntax of the flattening operator]\label{def:syntax-flattening-operator}
	Let $\varphi$ be a team-based $\sigma$-formula, then $\F\varphi$ is a team-based $\sigma$-formula. %
	The flattening operator does not bind variables, thus $\FV(\F\varphi)=\FV(\varphi)$.
\end{definition}

It should be noted that the flattening operator $\F$ may be applied to any kind of team-based formula. Let $\mathcal{D}=\left\{d_{1},\dots,d_{n}\right\}$ be a set of team-based constructors (e.g. dependency, anonymity, independency, inclusion, exclusion, etc.). %
We will refer to the extension of  $\FO(\mathcal{D})$ by the flattening operator $\F$ by $\FO(\mathcal{D},\F)$. 

\begin{definition}[Semantics of the flattening operator]
	Let $\varphi$ be a team-based $\sigma$-formula, $\mathfrak{M}$ a $\sigma$-model and $X$ a team over $\mathfrak{M}$ such that $\dom(X)\supseteq\FV(\varphi)$. We define the truth of the formula $\F\varphi$ in the model $\mathfrak{M}$ and the team $X$ as:
	\[
		\mathfrak{M}\vDash_{X}\F\varphi\iff\forall s\in X,\mathfrak{M}\vDash_{\left\{s\right\}}\varphi
	\]
\end{definition}

Therefore, one can rewrite the flatness definition as follows: let $\varphi$ be a team-based formula, $\mathfrak{M}$ a model and $X$ a team over $\mathfrak{M}$ such that $\dom\left(X\right)\supseteq\FV\left(\varphi\right)$ (also equal to $\FV\left(\F\varphi\right)$), then $\varphi$ is flat if and only if it holds that
\[
	\mathfrak{M}\vDash_{X}\varphi\iff\mathfrak{M}\vDash_{X}\F\varphi.
\]
It is precisely because of this characterization of flattenness that the flattening operator was introduced, and that it bears the name that it does. One can easily show the following immediate properties about the flattening operator.

\begin{property}\label{pty:flattening-operator-properties}
	Let $\varphi$ be any team-based formula. Then the following properties hold:
	\begin{enumerate}
		\item\label{pty:flattening-operator-properties:F-phi-flat} the formula $\F\varphi$ is always flat.
		
		\item\label{pty:flattening-operator-properties:flat-implies-F-phi-equiv-phi} $\varphi$ is flat if and only if $\F\varphi\equiv\varphi$.
		
		\item\label{pty:flattening-operator-properties:indempotency} the operator $\F$ is idempotent, that is, $\F\F\varphi\equiv\F\varphi$.
	\end{enumerate}
\end{property}

\begin{proof}
	\leavevmode
	\begin{enumerate}
		\item[\ref{pty:flattening-operator-properties:F-phi-flat}] Let $\mathfrak{M}$ be a model, and let $X$ be a team over $\mathfrak{M}$ such that $\dom\left(X\right)\supseteq\FV\left(\F\varphi\right)$. According to \cref{def:syntax-flattening-operator}, we get that $\FV\left(\F\varphi\right)=\FV\left(\varphi\right)$. Then
		\begin{align*}
			&\mathfrak{M}\vDash_{X}\F\varphi \\
			\iff&\forall s\in X,\mathfrak{M}\vDash_{\left\{s\right\}}\varphi & \text{(semantics of $\F$)} \\
			\iff&\forall s\in X,\left(\forall s'\in\left\{s\right\},\mathfrak{M}\vDash_{\left\{s'\right\}}\varphi\right) \\
			\iff&\forall s\in X,\mathfrak{M}\vDash_{\left\{s\right\}}\F\varphi
		\end{align*}
		
		\item[\ref{pty:flattening-operator-properties:flat-implies-F-phi-equiv-phi}] Let $\mathfrak{M}$ be a model, and let $X$ be a team over $\mathfrak{M}$ such that $\dom\left(X\right)\supseteq\FV\left(\varphi\right)$. Suppose that $\varphi$ is flat, then $\mathfrak{M}\vDash_{X}\F\varphi$ iff for all $s\in X$, $\mathfrak{M}\vDash_{\left\{s\right\}}\varphi$ iff $\mathfrak{M}\vDash_{X}\varphi$. Conversely, suppose that $\F\varphi\equiv\varphi$, then $\mathfrak{M}\vDash_{X}\varphi$ iff $\mathfrak{M}\vDash_{X}\F\varphi$ iff for all $s\in X$, $\mathfrak{M}\vDash_{\{s\}}\varphi$.
		
		\item[\ref{pty:flattening-operator-properties:indempotency}] This is a direct consequence of the two previous items.
	\end{enumerate}
\end{proof}

\begin{property}\label{pty:F-implies-fo}
	Let $\varphi$ be a team-based formula, $\mathfrak{M}$ a model and $X$ a team over $\mathfrak{M}$ such that $\dom\left(X\right)\supseteq\FV\left(\varphi\right)$. If $\varphi$ is a flat formula, then
	\[
		\mathfrak{M}\vDash_{X}\F\varphi\implies\mathfrak{M}\vDash_{X}\varphi^{\f}.
	\]
\end{property}

\begin{proof}
	Suppose that $\varphi$ is flat, then $\mathfrak{M}\vDash_{X}\F\varphi$ implies by \cref{pty:flattening-operator-properties:flat-implies-F-phi-equiv-phi} that $\mathfrak{M}\vDash_{X}\varphi$, which implies by the \hyperref[ax:entailment-axiom]{entailment axiom} that $\mathfrak{M}\vDash_{X}\varphi^{\f}$.
\end{proof}

\begin{remark}
	Caution: the reciprocal of \cref{pty:F-implies-fo} is false. Consider the $\ANON$-sentence $\varphi\coloneqq\forall x\exists y\left(x\anon y\right)$ and the model $\mathfrak{M}$ with domain $M=\left\{a\right\}$ consisting in only one element $a$. Then $\varphi^{\f}\equiv\exists y.\top$, and so $\mathfrak{M}\vDashteam\varphi^{\f}$. But $\mathfrak{M}\nvDashteam\varphi$, and as $\varphi$ is a sentence, then $\F\varphi\equiv\varphi$, and so $\mathfrak{M}\nvDashteam\F\varphi$. Hence $\varphi^{\f}\nvDashteam\F\varphi$.
	
	As a consequence, $\F\varphi$ needs not be logically equivalent to $\varphi^{\f}$, even for $\varphi$ that are already flat.
\end{remark}

\subsection{Basic properties of the flattening operator}

We start by stating some simple facts about the application of the flattening operator to the most well-known atoms of dependency. The proofs are straightforward and are thus omitted.

\begin{proposition}\label{prp:flattening-operator-atoms}
	\leavevmode
	\begin{multicols}{2}
		\begin{enumerate}
			\item\label{prp:flattening-operator-atoms:lit} for first-order literals $\alpha$: $\F\alpha\equiv\alpha$,
			
			\item\label{prp:flattening-operator-atoms:dep} for dependency atoms: $\F(\dep(\tuple{t}_{1},t_{2}))\equiv\top$,
			
			\item\label{prp:flattening-operator-atoms:ndep} for anonymity atoms: $\F(\tuple{t}_{1}\anon t_{2})\equiv\bot$,
			
			\item\label{prp:flattening-operator-atoms:incl} for inclusion atoms: $\F(\tuple{t}_{1}\incl\tuple{t}_{2})\equiv\tuple{t}_{1}\syntacticeq\tuple{t}_{2}$,
			
			\item\label{prp:flattening-operator-atoms:excl} for exclusion atoms: $\F(\tuple{t}_{1}\excl\tuple{t}_{2})\equiv\tuple{t}_{1}\syntacticneq\tuple{t}_{2}$,
			
			\item\label{prp:flattening-operator-atoms:ind} for independency atoms: $\F(\tuple{t}_{2}\ind_{\tuple{t}_{1}}\tuple{t}_{3})\equiv\top$.
		\end{enumerate}
	\end{multicols}
\end{proposition}

We now state simple properties about the distributivity of the flattening operator with respect to the Boolean connectives. 
The proof of the next proposition is a routine application of the semantics.
\begin{proposition}
For any team-based formulas $\varphi,\psi$ it holds that
\[
	\F(\varphi\land\psi)\equiv\F\varphi\land\F\psi.
\]
If, $\varphi,\psi$ are both satisfied by the empty team, then additionally
\[
	\F(\varphi\lor\psi)\equiv\F\varphi\lor\F\psi.
\]
\end{proposition}

When the logic does not have the empty team property, this proposition does not hold. A counter-example is provided by the formula $\NE\vee\NE$, where $\NE$ is satisfied by a team if the team is non-empty.

We continue by studying commutativity of the flattening operator with the  quantifiers.

\begin{definition}
	A formula $\varphi$ is called \defexpr{downwards flat} (DF) if
	\[
		\mathfrak{M}\vDash_{X}\varphi\implies\forall s\in X,\mathfrak{M}\vDash_{\left\{s\right\}}\varphi
	\]
	for all models $\mathfrak{M}$ and for all teams $X$ over $\mathfrak{M}$ such that $\dom(X)\supseteq\FV(\varphi)$. Similarly, $\varphi$ is called \defexpr{upwards flat} (UF) if
	\[
		(\forall s\in X,\mathfrak{M}\vDash_{\left\{s\right\}}\varphi)\implies\mathfrak{M}\vDash_{X}\varphi
	\]
	for all models $\mathfrak{M}$ and for all teams $X$ over $\mathfrak{M}$ such that $\dom(X)\supseteq\FV(\varphi)$.
\end{definition}

We sum-up the relations between the different flatness notions we introduced so far in the \cref{fig:sum-up-flatness-notions}.

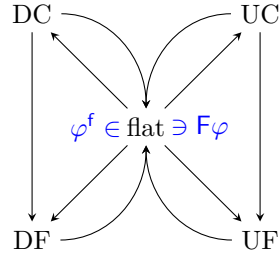
\begin{figure}
	\centering
	\begin{tikzpicture}
		\draw (0,0)   node (flat) {flat};
		\draw (-1.5,1.5)  node (dc)   {DC};
		\draw (1.5,1.5)   node (uc)   {UC};
		\draw (-1.5,-1.5) node (df)   {DF};
		\draw (1.5,-1.5)  node (uf)   {UF};
		
		\draw [->] (flat) -- (dc);
		\draw [->] (flat) -- (uc);
		\draw [->] (flat) -- (df);
		\draw [->] (flat) -- (uf);
		
		\draw [->] (dc) -- (df);
		\draw [->] (uc) -- (uf);
		
		\draw [->] (df) to [out=0,in=270] (flat);
		\draw [->] (uf) to [out=180,in=270] (flat);
		
		\draw [->] (dc) to [out=0,in=90] (flat);
		\draw [->] (uc) to [out=180,in=90] (flat);
		\draw [blue] (-0.2,0) node [left] {$\varphi^{\f}\in$};
		\draw [blue] (0.2,0) node [right] {$\ni\F\varphi$};
	\end{tikzpicture}
	\caption{Relation between all flatness notions}
	\label{fig:sum-up-flatness-notions}
\end{figure}

One can easily remark that the property of downwards closure is strictly  stronger assumption than downwards flatness. However, it will be a sufficient assumption for the following propositions.

\begin{proposition}
	\begin{enumerate}
		\item If
		$\exists x\varphi$ is downwards flat, then $\F(\exists x\varphi)\vDash\exists x\F\varphi$.
		
		\item If
		$\exists x\varphi$ is upwards flat, then $\exists x\F\varphi\vDash\F(\exists x\varphi)$.
	\end{enumerate}
\end{proposition}

\begin{proof} We consider the proof of the first claim.
	Let $\mathfrak{M}$ be a model, and let $X$ be a team over $\mathfrak{M}$ such that $\dom(X)\supseteq\FV(\varphi)\cup\FV(\psi)$. Then
	
	\begin{align*}
		&\mathfrak{M}\vDash_{X}\F(\exists x\varphi)\\
		\implies&\forall s\in X,\mathfrak{M}\vDash_{\left\{s\right\}}\exists x\varphi\\
		\implies&\forall s\in X,\exists H:\left\{s\right\}\to\powersetstar(M),\mathfrak{M}\vDash_{\left\{s\right\}\left[H\replaces x\right]}\varphi&\text{(semantics of $\exists$)}\\
		\implies&\exists H':X\to\powersetstar(M),\forall s\in X,\mathfrak{M}\vDash_{\left\{s\right\}\left[H'\replaces x\right]}\varphi\\
		\implies&\exists H':X\to\powersetstar(M),\forall s\in X,\forall s'\in\left\{s\right\}\left[H'\replaces x\right],\mathfrak{M}\vDash_{\left\{s'\right\}}\varphi&\text{(downwards flatness)}\\
		\implies&\exists H':X\to\powersetstar(M),\forall s\in X\left[H'\replaces x\right],\mathfrak{M}\vDash_{\left\{s\right\}}\varphi\\
		\implies&\exists H':X\to\powersetstar(M),\mathfrak{M}\vDash_{X\left[H'\replaces x\right]}\F\varphi\\
		\implies&\mathfrak{M}\vDash_{X}\exists x(\F\varphi)
	\end{align*}
\end{proof}

\begin{proposition}
	\begin{enumerate}
		\item if %
		$\forall x\varphi$ is downwards flat, then $\F(\forall x\varphi)\vDash\forall x(\F\varphi)$,
		
		\item if %
		$\forall x\varphi$ is upwards flat, then $\forall x(\F\varphi)\vDash\F(\forall x\varphi)$.
	\end{enumerate}
\end{proposition}

\begin{proof} Again we consider the proof of the first claim only as the proof is similar in the second case.
	Let $\mathfrak{M}$ be a model, and let $X$ be a team over $\mathfrak{M}$ such that $\dom(X)\supseteq\FV(\varphi)\cup\FV(\psi)$. Then
	
	\begin{align*}
		&\mathfrak{M}\vDash_{X}\F(\forall x\varphi)\\
		\implies&\forall s\in X,\mathfrak{M}\vDash_{\left\{s\right\}}\forall x\varphi\\
		\implies&\forall s\in X,\mathfrak{M}\vDash_{\left\{s\right\}\left[M\replaces x\right]}\varphi\\
		\implies&\forall s\in X,\forall s'\in\left\{s\right\}\left[M\replaces x\right],\mathfrak{M}\vDash_{\left\{s'\right\}}\varphi&\text{(downwards flatness)}\\
		\implies&\forall s\in X\left[M\replaces x\right],\mathfrak{M}\vDash_{\left\{s\right\}}\varphi\\
		\implies&\mathfrak{M}\vDash_{X\left[M\replaces x\right]}\F\varphi\\
		\implies&\mathfrak{M}\vDash_{X}\forall x(\F\varphi)
	\end{align*}
\end{proof}

\subsection{Preservation of closure properties}

The proof of the next lemma is a routine application of the semantics.
\begin{lemma}\label{lem:preservation-closure-properties}
	\begin{enumerate}
		\item If $\varphi$ is downwards closed (resp. union closed), then $\F\varphi$ is also downwards closed (resp. union closed).
		\item If $\varphi$ is downwards flat (resp. upwards flat), then $\F\varphi$ is also downwards flat (resp. upwards flat).
	\end{enumerate}
\end{lemma}

\begin{corollary}\label{cor:preservation-closure-properties-logics}
	Let $\mathcal{L}$ be team-based logic extending $\FO$ by some dependencies and operators.
	\begin{enumerate}
		\item  If $\mathcal{L}$ is downwards closed (resp. union closed, resp. downwards flat, resp. upwards flat), then so is $\mathcal{L}(\F)$.
		\item If $\mathcal{L}$ is flat, then $\mathcal{L}\equiv\mathcal{L}(\F)$.
	\end{enumerate}
\end{corollary}

\section{The expressive power of the flattening operator}
\label{sec:expressive-power-of-F}

We ask now, does the flattening operator, when added to a logic, increase the expressive power of the logic? In \cite{Galliani:2018} the flattening operator $\F$ was said to be \emph{safe} for a team-based logic $\mathcal{L}$ if  %
$\mathcal{L}\equiv\mathcal{L}(\F)$. 
The idea is that the operator $\F$ can be ``safely'' added to the logic, as it does not increase its expressive power. In this section we systematically study the effect of $\F$ to the expressive power of various team-based logics. %

\subsection{Adding \texorpdfstring{$\F$}{F} to inclusion and anonymity logic}
\label{subsec:safeness-of-F-for-incl}

The goal of this subsection is to show that inclusion logic together with the flattening operator $\FO\left(\incl,\F\right)$ is equivalent to the positive fragment of  greatest fixed-point logic $\GFPpos$. From this it follows that $\FO\left(\incl,\F\right)$ has exactly the same expressive power as $\FOincl$ by the results of \cite{Galliani:2013:inclusion}.%

\begin{theorem}\label{thm:GFPpos-expresses-FOinclF}
	For every $\FO\left(\incl,\F\right)$-formula $\varphi(\tuple{x})$ with free variables in $\tuple{x}=(x_{1},\dots,x_{n})$ there is a $\GFPpos$-formula $\varphi^{*}=\varphi^{*}(R,\tuple{x})$ such that $\ar(R)=n$, $R$ occurs only positively in $\varphi^{*}$, and the condition
	\[
		\mathfrak{M}\vDash_{X}\varphi(\tuple{x})\quad\text{iff}\quad\left<\mathfrak{M},R\coloneqq X(\tuple{x})\right>\vDash_{s}\varphi^{*}(R,\tuple{x})\text{ for all }s\in X
	\]
	holds for all models $\mathfrak{M}$ and all teams $X$ over $\mathfrak{M}$ with $\dom(X)\supseteq\left\{x_{1},\dots,x_{n}\right\}$.
\end{theorem}

\begin{proof}
	The proof is by structural induction on $\varphi$. The cases for first-order literals, inclusion  atoms, connectives, and quantifiers are treated as in \cite[Theorem 15]{Galliani:2013:inclusion}. We show the inductive case corresponding to $\F$.
	
		 If $\varphi(\tuple{x})$ is of the form $\F\psi(\tuple{x})$, then we define %
		
		\[
			\varphi^{*}(R,\tuple{x})=\forall\tuple{y}(\neg R\tuple{y}\vee \theta(\tuple{x},\tuple{y}))
		\]
		where $\theta(\tuple{x},\tuple{y})=\psi^{*}(S,\tuple{x})\left[(\tuple{y}=\tuple{t})\replaces S(\tuple{t})\right]$, that is, $\theta$ is the formula $\psi^{*}(S,\tuple{x})$ in which we replaced each atom of the form $S(\tuple{t})$ by the atom $\tuple{y}=\tuple{t}$. Let $\mathfrak{M}$ be a model and $X$ be a team over $\mathfrak{M}$ with $\dom(X)\supseteq\left\{x_{1},\dots,x_{n}\right\}$. One can remark that
		
		\begin{align}	\notag\mathfrak{M}\vDash_{X}\F\psi(\tuple{x})
			\iff&\forall s\in X.&&\mathfrak{M}\vDash_{\left\{s\right\}}\psi(\tuple{x})&\text{}\\
			\notag\iff&\forall s\in X.\forall s'\in\left\{s\right\}.&&\left<\mathfrak{M},S\coloneqq\left\{s\right\}(\tuple{x})\right>\vDash_{s'}\psi^{*}(S,\tuple{x})&\text{(by IH on $\psi(\tuple{x})$)}\\
			\notag\iff&\forall s\in X.&&\left<\mathfrak{M},S\coloneqq\left\{s\right\}(\tuple{x})\right>\vDash_{s}\psi^{*}(S,\tuple{x})\\
			\notag\iff&\forall s\in X.&&\left<\mathfrak{M},S\coloneqq\left\{s(\tuple{x})\right\}\right>\vDash_{s}\psi^{*}(S,\tuple{x})\\
			\label[condition]{cnd:GFPpos-expresses-FOinclF:proof-case-F}\iff&\forall s\in X.&&\mathfrak{M}\vDash_{s}\psi^{*}(S,\tuple{x})\left[(s(\tuple{x})=\tuple{t})\replaces S(\tuple{t})\right]
		\end{align}

		The last equivalence is justified by the following observation: as $S$ is assigned to a relation with a unique element $s(\tuple{x})$, having an atom on the form $S(\tuple{t})$ is $\psi^{*}$ is equivalent to have the atom $s(\tuple{x})=\tuple{t}$ instead (we do the test \quoexpr{by hand}), so we perform this substitution in $\psi^{*}$ while preserving equivalence.
		
		Now, suppose $\mathfrak{M}\vDash_{X}\F\psi(\tuple{x})$, then \cref{cnd:GFPpos-expresses-FOinclF:proof-case-F} holds. We want to show
		\[
			\left<\mathfrak{M},R\coloneqq X(\tuple{x})\right>\vDash_{s}\forall\tuple{y}(R\tuple{y}\rightarrow\psi^{*}(S,\tuple{x})\left[(\tuple{y}=\tuple{t})\replaces S(\tuple{t})\right])\quad\text{for all }s\in X
		\]
		To do so, let $\tuple{b}\in M$ such that $\tuple{b}\in X(\tuple{x})$. Then there is an assignment $s\in X$ such that $\tuple{b}=s(\tuple{x})$, we fix such an assignment. And by \cref{cnd:GFPpos-expresses-FOinclF:proof-case-F} we obtain
		\[
			\mathfrak{M}\vDash_{s}\psi^{*}(S,\tuple{x})\left[(\tuple{b}=\tuple{t})\replaces S(\tuple{t})\right]
		\]
		
		Conversely, suppose that
		\[
			\left<\mathfrak{M},R\coloneqq X(\tuple{x})\right>\vDash_{s}\forall\tuple{y}(R\tuple{y}\rightarrow\psi^{*}(S,\tuple{x})\left[(\tuple{y}=\tuple{t})\replaces S(\tuple{t})\right])\quad\text{for all }s\in X
		\]
		So by taking $\tuple{y}\coloneqq s(x)$, as $s(\tuple{x})\in\left\{s(x)\setsep s\in X\right\}\defeq X(\tuple{x})$, then the condition \extexpr{$R\tuple{y}$} is satisfied, so we get
		\[
			\mathfrak{M}\vDash_{s}\psi^{*}(S,\tuple{x})\left[(s(\tuple{x})=\tuple{t})\replaces S(\tuple{t})\right]\quad\text{for all }s\in X
		\]
		and the result follows from \cref{cnd:GFPpos-expresses-FOinclF:proof-case-F}.
\end{proof}
\begin{corollary}\label{cor:F-safe-inclusion-logic}\label{cor:F-safe-anonymity-logic}
	\leavevmode
	\begin{enumerate}
		\item The logics $\FO\left(\incl,\F\right)$ and $\FOincl$ are equivalent for sentences.
		\item The logics $\FO\left(\anon,\F\right)$ and $\FOanon$ are equivalent for sentences.
	\end{enumerate}
\end{corollary}

\begin{proof}
	The first claim follows by \cref{thm:GFPpos-expresses-FOinclF}.

	The second claim follows from the fact inclusion atoms can be expressed using anonymity atoms and vice versa.
\end{proof}

\subsection{Adding \texorpdfstring{$\F$}{F} to unary inclusion logic}

Galliani \cite{Galliani:2018} proved  that constancy atoms increase the expressive power of unary inclusion logic. More precisely, he proved that non-connectedness of graphs can be expressed  using constancy atoms and unary inclusion atoms, but not in terms of unary inclusion atoms alone. The next proposition shows that the same effect can be obtained using $\F$ and inclusion atoms. %

\begin{notation}
	Let $\phi$ is a first-order formula, and let $\psi$ be a team-based formula. We will then denote $\varphi\regularimplication\psi$ as an abbreviation for $\lnot\varphi\lor(\varphi\land\psi)$.
\end{notation}

We thank Pietro Galliani for suggesting the following proposition.

\begin{proposition}\label{prp:non-connectedness-expressible-FOincl1F}
	The following $\FO\left(\incl[1],\F\right)$ sentence
	\[
		\exists y(\F(\exists x(x\neq y\land\forall z(E(x,z)\regularimplication z\incl x))))
	\]
	is true in a graph $\mathfrak{G}=\left<V,E\right>$ if and only if the graph $\mathfrak{G}$ is disconnected.%
\end{proposition}

\begin{proof}
It is straightforward to check that the formula above is logically equivalent with the formula used in \cite{Galliani:2018} to show the same result.
\end{proof}

\begin{corollary}\label{cor:F-unsafe-unary-inclusion-logic}
	The flattening operator $\F$ increases the expressive power of $\FOincl[1]$.
\end{corollary}

\subsection{Adding \texorpdfstring{$\F$}{F} to unary anonymity logic}

\subsubsection{Expressivity of unary anonymity atoms and the flattening operator}

For any $n\in\mathbb{N}$, we denote by $\mathfrak{A}_{n}$ a graph consisting of two cycles each of length $2^{n+1}$ and $\mathfrak{B}_{n}$ a graph consisting of a single cycle of length $2^{n+2}$.

\begin{proposition}\label{separatingexmple}
	The following $\FO(\anon[1],\F)$ sentence
	\[
		\exists x\F(\exists y(y\neq x\land\forall z(E(y,z)\regularimplication(z\neq x\land y\anon z\land z\anon y))))
	\]
	is true in the models $\mathfrak{A}_{n}$ and is false in $\mathfrak{B}_{n}$.
\end{proposition}

\begin{proof}
	The intuition behind this $\FO\left(\anon[1],\F\right)$-sentence is as follows: consider a cycle of even length. We would like to fill it completely with variable assignments, in order to be able to detect if there exists another node in the graph that is not part of the considered cycle. To do so, we are going to use the mechanism of the anonymity atom to fill this cycle. More precisely, we will partition the cycle into two sets of nodes, that will be assigned to $y$ and $z$ respectively, in such a way that there will be an alternation between the values taken by $y$ and $z$ along the cycle. This is illustrated in the \cref{fig:anonymity-atom-filling-cycle-mechanism}.
	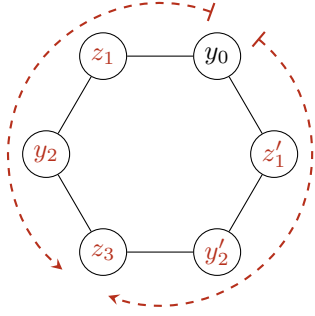
\begin{figure}
		\centering
		\begin{tikzpicture}[
			auto,
			node_style/.style={circle, draw},
			edge_style/.style={},
			]
			
			\def\radius{1.5cm}
			\def\n{6}
			
			\foreach \i in {1,...,\n} {
				\node [node_style] at ({360/\n * (\i - 1)}:\radius) (v\i) {$\phantom{y}$};
			}
			\foreach \i in {1,...,\n} {
				\pgfmathtruncatemacro{\next}{mod(\i, \n) + 1}
				\draw [edge_style]  (v\i) edge (v\next);
			}
			
			\node			at (v2) {$y_{0}$};
			\node [BrickRed] at (v4) {$y_{2}$};
			\node [BrickRed] at (v6) {$y_{2}'$};
			
			\node [BrickRed] at (v1) {$z_{1}'$};
			\node [BrickRed] at (v3) {$z_{1}$};
			\node [BrickRed] at (v5) {$z_{3}$};
			
			\def\radiusoffset{5mm}
			\def\angleoffset{10}
			\draw [thick, BrickRed, dashed, |->] (360/\n+\angleoffset:\radius+\radiusoffset) arc (360/\n * (2-1)+\angleoffset:360/\n * (5-1)-\angleoffset:\radius+\radiusoffset);
			\draw [thick, BrickRed, dashed, |->] (360/\n-\angleoffset:\radius+\radiusoffset) arc (360/\n * (2-1)-\angleoffset:360/\n * (5-1)+\angleoffset-360:\radius+\radiusoffset);
		\end{tikzpicture}
		\caption{Illustration of the anonymity atom mechanism to fill a cycle}
		\label{fig:anonymity-atom-filling-cycle-mechanism}
	\end{figure}		
	
	This filling process is done via the sub-formula $E(y,z)\regularimplication(y\anon z\land z\anon y)$. Here, the black node $y_{0}$ in the \cref{fig:anonymity-atom-filling-cycle-mechanism} corresponds to a starting value for the variable $y$. Then $E(y,z)\regularimplication y\anon z$ forces the variable $z$ to take two different values $z_{1}$ and $z_{1}'$ that are in fact the two adjacent vertices of $y_{0}$. Then $E(y,z)\regularimplication z\anon y$ firstly forces the variable $y$ to take two different values $y_{2}$ and $y_{0}$ that are in fact the two adjacent vertices of $z_{1}$, and secondly forces the variable $y$ to take two different values $y_{2}'$ and $y_{0}$ that are in fact the two adjacent vertices of $z_{1}'$. And so on. At the end, as the considered cycle is of even length, then the last two values assigned to the variable $z$ actually coincide. Such filled components then satisfy the clause $\forall z(E(y,z)\regularimplication y\anon z\land z\anon y)$.

	The whole objective will be to succeed in distinguishing the graphs $\mathfrak{A}_{n}$ from the graphs $\mathfrak{B}_{n}$. In fact, the graphs $\mathfrak{A}_{n}$ and $\mathfrak{B}_{n}$ are all composed of one or two cycles of even length.
	
	At the end of the filling process, we obtained a cycle filled with an alternation of values for the variables $z$ and $y$. In order to distinguish $\mathfrak{A}_{n}$ from $\mathfrak{B}_{n}$, we need to find an $x$ such that it does not belong to the alternation of $y$ and $z$. Hence, we need to check both that $y\neq x$ and $z\neq x$, which gives
	\[
		y\neq x\land\forall z(E(y,z)\regularimplication(z\neq x\land y\anon z\land z\anon y))
	\]
	The end of the construction of the desired formula is identical to the one we made for the unary inclusion logic together with the flattening operator $\F$, and at the end we finally get
	\[
		\exists x\F(\exists y(y\neq x\land\forall z(E(y,z)\regularimplication(z\neq x\land y\anon z\land z\anon y))))
	\]
	By construction, it is clear that this formula is true for the graphs $\mathfrak{A}_{n}$ but is false for the graphs $\mathfrak{B}_{n}$.
\end{proof}

\subsubsection{Unary anonymity logic does not express non-connectedness}

The next technical lemma will be used to show that the structures $\mathfrak{A}_{n}$  and $\mathfrak{B}_{n}$ from the previous section cannot be separated by formulas of unary anonymity logic.

Let $\mathfrak{M}$ be a model, let $X$ be a team over $\mathfrak{M}$, and let $\mathcal{F}\subseteq\Aut\left(\mathfrak{M}\right)$. Then we define the \defexpr{team closure of $X$ under $\mathcal{F}$}, written $\TeamCl_{\mathcal{F}}\left(X\right)$, as the set of all assignments obtained by applying all automorphisms from $\mathcal{F}$ to all assignments of $X$, that is
\[
	\TeamCl_{\mathcal{F}}\left(X\right)\defeq\left\{f\left(s\right)\setsep s\in X,f\in\mathcal{F}\right\}
\]
For $\mathcal{F}=\Aut\left(\mathfrak{M}\right)$, we simply write $\TeamCl\left(X\right)$ instead of $\TeamCl_{\Aut\left(\mathfrak{M}\right)}\left(X\right)$, that is
\[
	\TeamCl\left(X\right)\defeq\left\{f\left(s\right)\setsep s\in X,f\in\Aut\left(\mathfrak{M}\right)\right\}
\]

Recall that a \defexpr{magma} is a set $M$ equipped with a binary operation $\star$ that sends any two elements from $M$ to another element in $M$, and in this case we call $\star$ an \defexpr{internal binary operation}. We call a \defexpr{unitary magma} any magma whose internal binary operation admits a neutral element (which is then unique), that is, if there is $e\in M$ such that for all $m\in M$ we get $e\star m=m\star e=m$.

\begin{lemma}\label{magma}
	Let $\mathfrak{M}$ be a model. Let $\mathcal{F}\subseteq\Aut(\mathfrak{M})$ be a unitary magma such that for any two points $m_{1},m_{2}\in M$ there exists an automorphism $f\in\mathcal{F}$ such that $f(m_{1})=m_{1}$ and $f(m_{2})\neq m_{2}$. Then for all teams $X$ over $\mathfrak{M}$ such that $\TeamCl_{\mathcal{F}}(X)=X$ and all formulas $\varphi\in\FOanon[1]$ with free variables in $\dom(X)$ 
	we have that
	\[
		\mathfrak{M}\vDash_{X}\varphi\quad\text{iff}\quad\mathfrak{M}\vDash_{X}\varphi^{\f}	
	\]
\end{lemma}

\begin{proof}
		It suffices to show the right to left direction that is proved via structural induction. The proof for the connectives and quantifiers are similar to that of Lemma~21 in \cite{Galliani:2018}.

			Assume $\varphi=x_{1}\anon x_{2}$. %
			As $(x_{1}\anon x_{2})^{\f}=\top$ %
			we need to prove that $\mathfrak{M}\vDash_{X}x_{1}\anon x_{2}$ only from the hypothesis of the lemma, that is, whenever $X$ is a team whose domain contains the variables $x_{1}$ and $x_{2}$ and such that $X=\TeamCl_{\mathcal{F}}(X)$. But this is the case. Indeed, suppose that $s(x_{1})=m_{1}$ and $s(x_{2})=m_{2}$. Then by assumption of the lemma, there is an automorphism $f\in\mathcal{F}$ such that $f(m_{1})=m_{1}$ and $f(m_{2})\neq m_{2}$, and since $X=\TeamCl_{\mathcal{F}}(X)$ there exists some assignment $s'\in X$ such that $s'=f(s)$, that is, $s'(x)=f(s(x))$ for all $x\in\dom(s)$. This implies in particular that
			\[\begin{cases}
				s'(x_{1})=f(s(x_{1}))=f(m_{1})=m_{1}=s(x_{1})\\
				s'(x_{2})=f(s(x_{2}))=f(m_{2})\neq m_{2}=s(x_{2})
			\end{cases}\]
			and thus, for any assignment $s\in X$ there exists some assignment $s'\in X$ such that $s'(x_{1})=s(x_{1})$ but $s'(x_{2})\neq s(x_{2})$. This shows that $\mathfrak{M}\vDash_{X}x_{1}\anon x_{2}$, as required.
\end{proof}

\begin{corollary}\label{cor:F-unsafe-unary-anonymity-logic}
	The flattening operator increases the expressive power of unary anonymity logic.	
\end{corollary}

\begin{proof}
	By Proposition \ref{separatingexmple}, the formula $\Theta$ separates the models $\mathfrak{A}_{n}$ and $\mathfrak{B}_{n}$. Assume for a contradiction that $\Theta$ is equivalent to some $\varphi$ of unary anonymity logic. By Lemma \ref{magma}, $\varphi$ is equivalent to $\varphi^{\f}$ in the models $\mathfrak{A}_{n}$ and $\mathfrak{B}_{n}$, that is
	\[
		\mathfrak{A}_{n}\vDashTarski \varphi^{\f} \iff \mathfrak{B}_{n}\vDashTarski \varphi^{\f},
	\]
	where $\varphi^{\f}$ is a first-order formula. But by an Ehrenfeucht-Fraïssé game argument, this is impossible.
\end{proof}

\subsection{Adding \texorpdfstring{$\F$}{F} to dependence, exclusion and independence logic}

In this section we show that $\FO(\depformat,\F)$ has exactly the same expressive power as $\FOdep$. We will use the same strategy as in \cref{subsec:safeness-of-F-for-incl}. More precisely, we will show that $\FO(\depformat,\F)\subseteqq\ESO$, and as it is well known that $\ESO\equiv\FO(\depformat)$, we will obtain the result. The same idea works also  for exclusion and independence logic.

Exactly as for inclusion logic and $\GFPpos$ (\cref{thm:GFPpos-expresses-FOinclF}), we have to adapt the equivalence between $\FO(\depformat,\F)$ and $\ESO$ because they do not use the same semantics.

Recall that a sentence $\phi(R)$ from second-order logic is called \defexpr{downwards closed} if $(\mathfrak{M},A)\vDash\phi(R)$ and $A'\subseteq A$ imply $(\mathfrak{M},A')\models\phi(R)$.

\begin{lemma}\label{lem:translation-FOdepF-to-ESO}
	For every $\FO(\depformat,\F)$-formula $\varphi(\tuple{x})$ with the free variables in $\tuple{x}=x_{1},\dots,x_{n}$, there is an $\ESO$-formula $\varphi^{*}\equiv\varphi^{*}(R)$ which is downwards closed and such that $\ar(R)=n$, and that satisfies the following condition:
	\begin{equation}\label{eqn:translation-FOdepF-to-ESO}
		\mathfrak{M}\vDash_{X}\varphi(\tuple{x})\quad\text{iff}\quad\left<\mathfrak{M},R\coloneqq X(\tuple{x})\right>\vDash
\varphi^{*}(R)\vee\forall \tuple{x}\neg R(\tuple{x})
	\end{equation}
	holds for all models $\mathfrak{M}$ and all teams $X$ over $\mathfrak{M}$ with $\dom(X)\supseteq\left\{\tuple{x}\right\}$.
\end{lemma}
\begin{proof}
	 We consider the inductive case corresponding to $\varphi=\F\psi$ as the other cases can be dealt with analogously to \cite[Theorem 6.2]{Vaananen:2007:dependence}. It turns out that the formula $\varphi^{*}$ that we need to take is exactly the same as for the translation from $\FO(\incl,\F)$ to $\GFPpos$ which is clearly in $\ESO$ assuming $\psi^{*}$ is.
\end{proof}
 
\begin{corollary}\label{prp:F-safe-exclusion-logic}\label{prp:F-safe-dependency-logic}
	The operator $\F$ does not increase the expressive power of dependence and exclusion logic.	
\end{corollary}

\begin{proof}
	This follows by \cref{lem:translation-FOdepF-to-ESO} as any downwards closed  $\ESO$-formula $\varphi(R)$ can be translated into a dependence logic formula satisfying \eqref{eqn:translation-FOdepF-to-ESO}, see \cite{Kontinen:2009,Kontinen:2011}. The claim for exclusion logic follows from the fact that exclusion atoms can be  expressed using dependence atoms and vice versa.
\end{proof}  

\begin{corollary}\label{cor:F-safe-independency-logic}
	The operator $\F$ does not increase the expressive power of independence  logic.	
\end{corollary}

\begin{proof}
	The proof is analogous to the case of dependence logic using the analogue of \cref{lem:translation-FOdepF-to-ESO} (without the assumption of downward closure) ans its converse shown in \cite{Galliani:2012}.
\end{proof}   

\subsection{Adding \texorpdfstring{$\F$}{F} to fragments of dependence logic}

We begin by showing that the flattening operator reduces to the flattening of a formula for existential dependence logic formulas. 

\begin{proposition}\label{existentialDEP}
	If $\varphi\in\FOdep$ is existential %
	then $\F\varphi\equiv\varphi^{\f}$.
\end{proposition}

\begin{proof} Recall that  $\mathfrak{M}\vDash_{X}\F\varphi$ iff $\mathfrak{M}\vDash_{\{s\}}\varphi$, for all $s\in X$. Now it is easy to show using induction on $\varphi$ using downwards closure (and singleton-valued supplements functions) that $\mathfrak{M}\vDash_{\{s\}}\varphi$ iff $\mathfrak{M}\vDash_{\{s\}}\varphi^f$ for any singleton team $\{s\}$, from which the claim follows. 
\end{proof}

Next we turn to the so-called constancy logic the formulas of which has the following normal form.

\begin{proposition}[\cite{Galliani:2012}]%
	\label{prp:constancy-logic-normal-form}%
	Let $\varphi(\tuple{x})\in\FOcons$. Then $\varphi$ is logically equivalent to a constancy logic formula of the form
	\[
		\exists\tuple{y}\left(\bigwedge_{i}\dep(y_{i})\land\alpha(\tuple{x},\tuple{y})\right)
	\]
	for some first-order formula $\alpha$.
\end{proposition}

\begin{proposition}\label{prp:F-safe-constancy-logic}\label{prp:F-safe-zero-anonymity-logic}
		The operator $\F$ does not increase the expressive power of constancy logic, that is $\FO(\consformat,\F)\equiv\FOcons$.

\end{proposition}

\begin{proof}
We prove the claim using induction on $\varphi \in \FO(\consformat,\F)$. We consider only the case where $\varphi=\F \psi$. By the induction hypothesis and the previous lemma, 
we get that $\psi$ is equivalent with some formula in the normal form:
	\[
		\exists\tuple{y}\left(\bigwedge_{i}\dep(y_{i})\land\alpha(\tuple{x},\tuple{y})\right)
	\]
	for some first-order formula $\alpha$. It now follows from Proposition \ref{existentialDEP}  that $\varphi\equiv \exists\tuple{y}\alpha$.
\end{proof} 

\subsection{Adding \texorpdfstring{$\F$}{F} to \texorpdfstring{$\FO$}{FO} with the Boolean negation}

In this section we consider the effects of adding $\F$ to $\FOneg$ (also denoted by $\NEG$ in the following).

\begin{definition}
	Let $\varphi$ be a $\sigma$-formula, $\mathfrak{M}$ a $\sigma$-model and $X$ a team over $\mathfrak{M}$ such that $\dom(X)\supseteq\FV(\varphi)$. We define the truth of the formula $\cneg\varphi$ (the Boolean negation of $\varphi$) in the model $\mathfrak{M}$ and the team $X$ as:
	\[
		\mathfrak{M}\vDash_{X}\cneg\varphi\iff\mathfrak{M}\nvDash_{X}\varphi
	\]
\end{definition}

We will use the following notations from \cite{Luck:2020}:
\begin{itemize}
	\item $\E\varphi\defeq\cneg\neg\varphi$, meaning  that at least one assignment in the considered team satisfies $\varphi$;
	
	\item $\varphi\ovee\psi\defeq\cneg(\cneg\varphi\land\cneg\psi)$, so that $\mathfrak{M}\vDash_{X}\varphi\ovee\psi$ iff $\mathfrak{M}\vDash_{X}\varphi$ or $\mathfrak{M}\vDash_{X}\psi$, which corresponds to the  Boolean disjunction (a.k.a. intuitionistic disjunction). 
\end{itemize}
We will utilize the following normal form theorem.
\begin{proposition}[\cite{Luck:2020}]
	\label{prp:negation-logic-normal-form}%
	Let $\varphi(\tuple{x})\in\FOneg$. Then $\varphi$ is logically equivalent to a formula of the form
	\[
		\bigovee_{i\in I}\left(\alpha_{i}\land\bigwedge_{j\in J_{i}}\E\beta_{i,j}\right)
	\]
	for some finite sets $I$ and $J_{i}$ and for first-order formulas $\alpha_{i}$ and $\beta_{i,j}$ for all $i\in I$ and $j\in J_{i}$.
\end{proposition}

\begin{proposition}\label{prp:F-safe-negation-logic}
	$\FO(\cneg,\F)\equiv\FOneg$.
\end{proposition}

\begin{proof}
	To prove the non-trivial direction ($\Rightarrow$), we need to show that for all $\varphi(\tuple{x})\in\FO(\cneg,\F)$, there is $\varphi^{*}(\tuple{x})\in\FOneg$ such that for all models $\mathfrak{M}$ and for all teams $X$ over $\mathfrak{M}$ such that $\dom(X)\supseteq\left\{\tuple{x}\right\}$, the following holds
	\[
\mathfrak{M}\vDash_{X}\varphi(\tuple{x})\iff\mathfrak{M}\vDash_{x}\varphi^{*}(\tuple{x})
	\]
	We prove it by structural induction on $\varphi$. The only interesting case is when $\varphi(\tuple{x})=\F\psi(\tuple{x})$.
	
	According to \cref{prp:negation-logic-normal-form}, we get that
	\[
	   \psi(\tuple{x})\equiv\bigovee_{i\in I}\left(\alpha_{i}\land\bigwedge_{j\in J_{i}}\E\beta_{i,j}\right)
	\]
	for some finite sets $I$ and $J_{i}$ and for first-order formulas $\alpha_{i}$ and $\beta_{i,j}$ for all $i\in I$ and $j\in J_{i}$. It is now a routine application of the semantics to show that $\varphi$ is equivalent with the first-order formula 
	\[
	   \bigvee_{i\in I}\left(\alpha_{i}\land\bigwedge_{j\in J_{i}}\beta_{i,j}\right),
	\]
	using the facts that in singleton teams $\E\beta$ is equivalent with $\beta$ and that for all $s\in X$
	\[
		\mathfrak{M}\vDash_{\{s\}} \bigovee_i \varphi_i\quad\text{iff}\quad\mathfrak{M}\models_{X} \bigvee_i \varphi_i,
	\]
	for  first-order formulas $\varphi_i$.
\end{proof}

\section{Summing up expressive powers}

We sum up all the previous results in the following \cref{fig:sum-up}. It shows when a fragment is comparable to another in terms of expressive power, and when it is the case, it indicates which fragment is more expressive than another. The red edges are the comparisons that are not known at the moment, but which we believe to be true. These questions have some echoes in finite model theory where beyond the seminal and old results about the monadic case (i.e. arity $1$) of existential second order logic, very few separation results are known~\cite{Ajtai:1983,Ajtai:2000,Durand:1998}. Considering the tight connections between arity based fragments of existential second order logic and of some team logics~\cite{Durand:2012}, such a situation is all but surprising.

\begin{figure}
	\centering
	\resizebox{\textwidth}{!}{
	\begin{tikzpicture}
		\def\rrrpos{7.5}
		\def\rrpos {5}
		\def\rpos  {1.5}
		\def\lpos  {-\rpos}
		\def\llpos {-\rrpos}
		\def\lllpos{-\rrrpos}
		
		\node [rounded corners, draw] at (0,0)	   (FO)	   {$\FO\left(\F\right)\overset{\text{\crefabbr{cor:preservation-closure-properties-logics}}}{\equiv}\FO$};

		\node [rounded corners, draw] at (0,1.5)	   (FOneg)	   {$\FO\left(\sim,\F\right)\overset{\text{\crefabbr{prp:F-safe-negation-logic}}}{\equiv}\FO\left(\sim\right)$};
		
		\node [rounded corners, draw] at (\llpos,1.5) (FOanon0)  {$\FOanon[0]\textcolor{red}{\equiv?\FO\left(\anon[0],\F\right)}$};
		\node [rounded corners, draw] at (\llpos,3) (FOincl1)  {$\FOincl[1]$};
		\node [rounded corners, draw] at (\lllpos,3.75)   (FOincl1F) {$\FO\left(\incl[1],\F\right)$};
		\node [rounded corners, draw] at (\lpos,3.75)   (FOincl1cons) {$\FO\left(\incl[1],\consformat\right)$};
		\node [rounded corners, draw] at (\llpos,4.5) (FOanon1)  {$\FOanon[1]$};
		\node [rounded corners, draw] at (\lllpos,5.25) (FOanon1F) {$\FO\left(\anon[1],\F\right)$};
		\node [rounded corners, draw] at (\llpos,6) (FOincl2)  {$\FOincl[2]$};
		\node [rounded corners, draw, red] at (\lllpos,6.75)   (FOincl2F) {$\FO\left(\incl[2],\F\right)$};
		\node [rounded corners, draw, align=center] at (\llpos,8) (FOincl)   {$\FO\left(\incl,\F\right)\overset{\text{\crefabbr{cor:F-safe-inclusion-logic}}}{\equiv}\FO\left(\incl\right)$\\$\equiv\FOanon\overset{\text{\crefabbr{cor:F-safe-anonymity-logic}}}{\equiv}\FO\left(\anon,\F\right)$};
		
		\node [rounded corners, draw] at (\rrpos,1.5) (FOcons)   {$\FO\left(\consformat,\F\right)\overset{\text{\crefabbr{prp:F-safe-constancy-logic}}}{\equiv}\FO\left(\consformat\right)$};
		\node [rounded corners, draw] at (\rrpos,3) (FOexcl1)  {$\FOexcl[1]$};
		\node [rounded corners, draw, red] at (\rrrpos,3.75) (FOexcl1F)  {$\FO\left(\excl[1],\F\right)$};
		\node [rounded corners, draw] at (\rrpos,4.5) (FOdep1)  {$\FOdep[1]$};
		\node [rounded corners, draw, red] at (\rrrpos,5.25) (FOdep1F) {$\FO\left(\depformat[1],\F\right)$};
		\node [rounded corners, draw] at (\rrpos,6) (FOexcl2)  {$\FOexcl[2]$};
		\node [rounded corners, draw, red] at (\rrrpos,6.75)   (FOexcl2F) {$\FO\left(\excl[2],\F\right)$};
		\node [rounded corners, draw, align=center] at (\rrpos,8) (FOdep)	{$\FO\left(\depformat,\F\right)\overset{\text{\crefabbr{prp:F-safe-dependency-logic}}}{\equiv}\FO\left(\depformat\right)$\\$\equiv\FO\left(\excl\right)\overset{\text{\crefabbr{prp:F-safe-exclusion-logic}}}{\equiv}\FO\left(\excl,\F\right)$};

		\node [rounded corners, draw] at (0,9.5)	   (FOind)	   {$\FO\left(\ind,\F\right)\overset{\text{\crefabbr{cor:F-safe-independency-logic}}}{\equiv}\FO\left(\ind\right)$};
		
		\draw (FO)	   -- (FOneg) node [midway, below, sloped] {$<$};
		
		\draw (FO)	   -- (\llpos,0) -- (FOanon0)  node [midway, below, sloped] {$<$};
		\draw (FOanon0)  -- (FOincl1)  node [midway, below, sloped] {$<$};
		\draw (FOincl1)  -- (FOincl1F) node [midway, below, sloped] {$>$} node [midway, above, sloped, font=\scriptsize] {\crefabbr{cor:F-unsafe-unary-inclusion-logic}};
		\draw (FOincl1)  -- (FOincl1cons) node [midway, below, sloped] {$<$} node [midway, above, sloped, font=\scriptsize] {\cite{Galliani:2018}};
		\draw (FOincl1)  -- (FOanon1)  node [midway, below, sloped] {$\leq$};
		\draw (FOanon1)  -- (FOanon1F) node [midway, below, sloped] {$>$} node [midway, above, sloped, font=\scriptsize] {\crefabbr{cor:F-unsafe-unary-anonymity-logic}};
		\draw (FOanon1)  -- (FOincl2)   node [midway, below, sloped] {$\leq$};
		\draw [red] (FOincl2)  -- (FOincl2F) node [midway, below, sloped] {$>?$};
		\draw (FOincl2)  -- (FOincl)   node [midway, below, sloped] {$<$} node [midway, sloped, fill=white] {\dots};
		
		\draw (FO)	   -- (\rrpos,0) -- (FOcons)   node [midway, below, sloped] {$<$};
		\draw (FOcons)   -- (FOexcl1)  node [midway, below, sloped] {$<$};
		\draw [red] (FOexcl1)  -- (FOexcl1F) node [midway, below, sloped] {$<?$};
		\draw (FOexcl1)  -- (FOdep1)   node [midway, below, sloped] {$\leq$};
		\draw [red] (FOdep1)  -- (FOdep1F) node [midway, below, sloped] {$<?$};
		\draw (FOdep1)  -- (FOexcl2)   node [midway, below, sloped] {$\leq$};
		\draw [red] (FOexcl2)  -- (FOexcl2F) node [midway, below, sloped] {$<?$};
		\draw (FOexcl2)  -- (FOdep)   node [midway, below, sloped] {$<$} node [midway, sloped, fill=white] {\dots};

		\draw (FOincl.north)   -- (FOind.west)  node [midway, above, sloped] {$<$};
		\draw (FOdep.north)	-- (FOind.east)  node [midway, above, sloped] {$>$};
	\end{tikzpicture}
	}
	\caption{Expressive power of logics on the level of formulas}
	\label{fig:sum-up}
\end{figure}

\section*{Conclusion}

We have explored the flattening operator $\F$ within the framework of team semantics, a powerful extension of classical first-order logic that allows for a more powerful treatment of dependencies and independencies in logical formulas. The flattening operator was originally introduced by Wilfrid Hodges in 1997 \cite{Hodges:1997:some}, but its potential and implications have remained largely unexplored until now. %

Through our investigation, we demonstrated that the introduction of the flattening operator $\F$ does not add expressive power to logics such as dependence logic, inclusion logic, constancy logic, inconstancy logic and $\FOneg$. It turns out that this is not the case with the fixed arity fragments of these logics. For example, we have seen that unary inclusion logic with the flattening operator $\FO\left(\incl[1],\F\right)$ is strictly more expressive than unary inclusion logic $\FOincl[1]$. We conjecture that the same is true of all the other fragments of fixed arity, too.

The work presented here opens several avenues for future research. One possible direction is the further exploration of the computational complexity associated with the translation of some logic fragments augmented with the flattening operator into $\SAT$ problems, along the lines of  \cite{Durand:2022:modular}. %

\section*{Acknowledgments}

The second author was supported by The Research Council of Finland, grant N°345634. The fourth author was supported by The Research Council of Finland, grant N°322795, and the European Research Council (ERC), grant agreement N°101020762.   

\printbibliography[heading=bibintoc]

\end{document}